\documentclass[%
 reprint,
 amsmath,amssymb,
 aps,
]{revtex4-2}

\RequirePackage{lineno}
\usepackage[pdftex]{graphicx}
\usepackage{wrapfig}
\usepackage{colortbl}
\usepackage[export]{adjustbox}
\usepackage[colorlinks=true,linkcolor=blue,citecolor=blue]{hyperref}

\usepackage{fancyvrb}
\usepackage{fixltx2e}
\graphicspath{{figures/}}
\DeclareGraphicsExtensions{.pdf,.jpeg,.png}
\usepackage{amsmath}
\usepackage{amsthm}

\usepackage{url}
\usepackage{comment}

\newtheorem{lemma}{Lemma}[section]
\newtheorem{corollary}{Corollary}[lemma]

\theoremstyle{definition}
\newtheorem{definition}{Definition}[section]

\newtheoremstyle{note}% 〈name〉
{3pt}% 〈Space above〉1
{3pt}% 〈Space below 〉1
{}% 〈Body font〉
{}% 〈Indent amount〉2
{\textbf{\itshape}}% 〈Theorem head font〉
{:}% 〈Punctuation after theorem head 〉
{.5em}% 〈Space after theorem head 〉3
{}% 〈Theorem head spec (can be left empty, meaning ‘normal’ )〉
\newtheorem{rem}{Remark}[section]
\usepackage{etoolbox}
\usepackage{physics}
\usepackage{amsmath}
\usepackage{tikz}
\usepackage{mathdots}
\usepackage{yhmath}
\usepackage{cancel}
\usepackage{color}
\usepackage{siunitx}
\usepackage{array}
\usepackage{multirow}
\usepackage{amssymb}
\usepackage{gensymb}
\usepackage{tabularx}
\usepackage{extarrows}
\usepackage{booktabs}
\usetikzlibrary{fadings}
\usetikzlibrary{patterns}
\usetikzlibrary{shadows.blur}
\usetikzlibrary{shapes}
\usepackage{xcolor}
\usepackage{ragged2e}
\usepackage[ruled,vlined,nofillcomment, linesnumbered]{algorithm2e}
\usepackage{xcolor}
\usepackage{bbold}
\usepackage{algcompatible}
\usepackage{subcaption}
\usepackage{caption}
\usepackage{efbox}

\SetKwProg{Fn}{Function}{}{}

\newcommand*{\Comb}[2]{{}^{#1}C_{#2}}%
\definecolor{Gray}{gray}{0.9}
\definecolor{LightCyan}{rgb}{0.88,1,1}

\begin{document}

\preprint{APS/123-QED}

\title{Unified Framework for Complex Graph-Data: Introducing the Hybrid Layered Network Model}% Force line breaks with \\

\author{Shraban Kumar Chatterjee}
\email{chatterjee.2@iitj.ac.in}
\affiliation{IIT Jodhpur}
 \altaffiliation[]{%
     Department
    of Computer Science and Engineering \\Indian Institute of Technology, Jodhpur, 342030.
    }
\author{Suman Kundu}%
 \email{suman@iitj.ac.in}
 \affiliation{IIT Jodhpur}
%

% \collaboration{MUSO Collaboration}%\noaffiliation

\date{\today}% It is always \today, today,

\begin{abstract}
The present paper provides a generalized model of network, namely, Hybrid Layered Network (HLN). We proved that the sets of all homogeneous, heterogeneous and multi-layered networks are subsets of the set of all HLNs depicting the model's generalizability. The proposed HLN is more efficient in encoding different types of nodes and edges {when compared to representing the same information through heterogeneous or multilayered networks}. It is found experimentally that the HLN model when used with GNNs improve tasks such as link prediction. In addition, we present a novel parameterized algorithm (with complexity analysis) for generating synthetic HLNs. The networks generated from our proposed algorithm are more consistent in modelling the layer-wise degree distribution of a real-world Twitter network (represented as HLN) than those generated by existing models. Moreover, we also show that our algorithm is capable of generating various multilayer and homogeneous network. Further, we define different structural measures for HLN {namely multilayer neighborhood, degree centrality, closeness centrality and betweeness centrality}. Accordingly, we established the equivalency of the proposed structural measures of HLNs with that of homogeneous, heterogeneous, and multi-layered networks.
\end{abstract}

\maketitle

%\tableofcontents

\section{Introduction}
\label{Intro}
Graph or Network data structure has long history in analyzing different complex systems, ranging from physical processes to biological systems. Networks are the root of fields like social network analysis and network sciences. In its basic structure, a graph or network consists of homogeneous nodes (or vertex) and connections (or edges) between them. Nodes represents entities like person, place, organization, hashtag, tweet etc. and connections represents relationships like friendship, following, works in etc. In many cases, complex systems cannot be expressed with this simple form. For example, in a single network there could be different type of nodes with diverse connections. Heterogeneous network model has been used for such networks \cite{nr, Pahwa2014, doi:10.1126/sciadv.aba4508}. In the literature, multilayer networks have also been used for different network science problems \cite{Berlingerio2013, Cardillo20132, 10.1371/journal.pcbi.1001106, 5992619}. 
However, these existing models fails to capture all the properties of the network effectively if the network contains various types of nodes with multiple relations between the same pair of nodes. These type of relationships are very common with modern online applications. For example, in X (formally twitter) an user can follow another user, mention the same user in post and reply to the same user for another posts. Now with this simple and frequent situation, a user is connected with another user with following-follower as well as mention relation. Similarly, an user is connected to a post with reply and mention relations. Hence, problems deals with current online social networks require to capture several relationships and entities simultaneously not just a friendship relation (in Facebook) or following-followers (in X). This is further exaggerated with the concept of decentralized social network introduced with the emergence of Web 3.0; ActivityPub protocol \cite{ActivityPub} is one step towards the next generation online social networking. Similarly, advance analysis of physical, biological, and chemical process now includes many more parameters than earlier that require more than a heterogeneous or multi-layer networks. All these require a new model for network analysis capable of equally handling simple networks and demanding networks such as decentralize social network and multifeatured social networks.  

\paragraph{Motivational Examples} (Example 1) Let us consider Facebook \cite{LEWIS2008330} network, it contains several types of nodes viz users, posts, pictures, and groups. All these entities are involved in many different interactions, e.g., friendship interactions, users tagged in a picture (may or may not be friend), user liked a post/picture, group related activities etc. A layer in Facebook can contain interactions between users based on friendship; another layer can contain relationships between users who belong to the same picture; and a third layer can be formed with interactions between users and groups. A multi-layered network cannot support heterogeneity in a layer due to the absence of node or edge types in their definition \cite{BOCCALETTI20141, Bianconi_2017, Bianconi2018, 10.1093/comnet/cnu016}. On the other hand, a single heterogeneous network cannot retain all the information present in a multi-layered network. The existing heterogeneous networks allow only one type of link between two objects. If we use existing data structures for the Facebook network, we will lose certain information.

(Example 2) Let us now look into a network with physical process. The network of chemical, gene, pathways, and diseases (CGPD) shows multilayer and heterogeneous characteristics \cite{Yang2022} similar to our example of Facebook. In the CGPD network, multiple types of nodes exist, such as chemicals, genes, biological pathways, and diseases. These entities engage in diverse interactions. For instance a layer may represent interactions between genes based on co-expression or shared functional annotations. Another layer could capture relationships between chemicals and genes, such as chemicals regulating gene expression or inhibiting protein function. A third layer may consist of interactions between pathways and diseases, reflecting how certain pathways contribute to disease progression. Additionally, direct interactions between chemicals and diseases (e.g., a drug targeting a specific disease) or genes linked to diseases (through genetic mutations or associations) create further connectivity.

(Example 3) Consider the decentralized social network platform Fediverse\footnote{https://fediverse.to}. This platform utilize the ActivityPub protocol and allow different social networking servers to communicate with each other. Each of these servers, running on different software, are of heterogeneous by itself providing services for blog, image sharing, video streaming, microblogging etc. to its users. In order to run network analysis the data structure must capture properties of each of the servers while preserving its characteristics. That is each site are the layers with heterogeneous nodes and links in it. 

\paragraph{Existing Approaches} Many networks used in different applications \cite{Malek2020, Hammoud2020, PhysRevE.102.032303, Hristova2016} are not homogeneous in nature. Some examples are already provided above. Nevertheless, such networks are assumed to be homogeneous \cite{6413734, 7288742}, heterogeneous \cite{Jin2021, 7033447} as well as multilayered \cite{Huang2021, PhysRevE.98.012303, PhysRevE.89.062817} network while solving different problems. Although the combination of the words `multilayer' and `heterogeneous' is used in many of the social network research \cite{gyanendro-singh-etal-2020-sentiment, RYBALOVA2021110477, 9454288, 8727718, Liu2024, 9756648, 10531187, 8683574, Pio-Lopez2021}, these methods are merely a manifestation of the multilayer network data structure \cite{BOCCALETTI20141} in the context of application. In other words, the underlying data structure used therein does not support heterogeniety in independent layer.
%Only a few works have tried to address heterogeneous and multi-layered properties together.
For example, work on sentiment analysis \cite{gyanendro-singh-etal-2020-sentiment} and inter-layer coupling dynamics \cite{RYBALOVA2021110477} consider each layer as homogeneous while different layers contain different types of nodes, i.e., none of the layers are heterogeneous. Nevertheless, the literature does not propose any unified model of network that incorporates all possible characteristics of modern complex networks.

\paragraph{Our Contribution} This paper proposes a new network model, Hybrid Layered Network (HLN) that unifies heterogeneous and multi-layered networks into one framework that can express modern complex networks, by supporting heterogeneity and multi-layered properties simultaneously. Homogeneous, heterogeneous and multilayered networks are special case of the proposed HLN. Various structural measures are developed for this network. In addition, the paper proposes a novel parameterized algorithm for generating a synthetic HLN. The algorithm is capable of generating homogeneous, heterogeneous, multilayered, and heterogeneous multilayered networks by setting the parameters appropriately. The paper has four main contributions as summarized below.

\begin{itemize}
    \item Proposes an unified model of network HLN. We define various structural measures for this model.
    \item We prove that the set of all homogeneous, heterogeneous, and multilayered networks is a subset of the set of all Hybrid Layered Networks.
    \item We present an algorithm that generates a HLN with various layers and different types of nodes.
    \item Various experimental results show the applicability of the proposed model in different applications and the benefit of unifying layers and heterogeneity within the model.
\end{itemize}
    
The remaining paper is organized as follows. In Section \ref{related} we briefly discuss the preliminaries and report the related work in the field. The proposed definition of a generalized hybrid layered network is presented in Section \ref{definition} and its structural properties are presented in Section \ref{neighborhood}. Section \ref{algorithm} contains the algorithm for generating a hybrid layered network and experimental results are presented in Section \ref{results}. Finally, Section \ref{conclude} concludes the findings.

\section{Related Work}
\label{related}
Graph has been in use for solving various real-life problems since the work of Euler on Konigsberg Bridge Problem in 1736. Graph used until 2009 are homogeneous that can be mathematically defined as  

\begin{definition}[Homogeneous Networks]\label{defHomo}
A homogeneous network is a graph $G = (V, E)$ with the vertex set $V$ and the edge set $E$ denoting the relations among these vertices.
\end{definition}

Heterogeneous networks are developed in late 2000s with the work of \cite{Sun2009RankClusIC, Sun2009RankingbasedCO, 10.1007/978-3-642-04747-3_2}. The survey \cite{7536145} presents a good idea about some of the more recent works on heterogeneous networks. Mathematically, heterogeneous network is defined as,

\begin{definition}[Heterogeneous Network \cite{7536145}]\label{defHet} A network $$H=(V, E, \{A, B\}, \{f_1, f_2\})$$ with edges having multiple nodes and edge types with functions $f_1$ and $f_2$ to map nodes and edges respectively to their types $A$ and $B$ is called a heterogeneous network. 
\end{definition}

Note that it is mandatory for either the node type or the edge type to be greater than one. Two links which belong to the same relation type have the same starting node type as well as the ending node type.
Heterogeneous networks are seen to be applied in link prediction \cite{7033447}, community detection \cite{9527702}, modeling human collective behavior \cite{7586100}, rail transit network \cite{9357955} and heterogeneous susceptible infected network \cite{10.1145/3485447.3512184}.
Another framework for network is multi-layer network. One of the pioneering works on multi-layer networks was by Moreno et al. \cite{10.1093/comnet/cnu016}. They have proposed a formal definition of multi-layered networks. The definition was further simplified along with the addition of structural measures in \cite{BOCCALETTI20141} as defined next. 

\begin{definition}[Multi-layered Network \cite{BOCCALETTI20141}]\label{defMultiOld}
A multi-layered network is defined as a triple $M = ( Y, G_{intra}, G_{inter})$ where $Y = \{1, 2, \cdots, k\}$ is the set of layers, $G_{intra} = (G_1, G_2, G_3, ..., G_k)$ is a sequence of graphs with each graph $G_i = (V_i, E_i)$   belonging to a layer, and $G_{inter} = \{G_{ij} = (V_i, V_j, E_{ij}) | i \neq j\}$. An inter layer graph $G_{ij}$ for layer $i$ to $j$ contains all the nodes and edges from layer $i$ to layer $j$.
\end{definition}

The formalism of multi-layer networks have led to various studies and applications on them. Multi-layer networks have been studied in various contexts like the study of flow processes or diffusion \cite{Cencetti_2019, Chang2019}, epidemic modelling and disease spreading \cite{PhysRevE.98.012303, Echegoyen_2021}, generalization of the percolation theory \cite{PhysRevResearch.2.033122, Zhang_2019}, clique based heuristic node analysis \cite{K2021}, localization properties of the network helping to understand the propagation of perturbation \cite{PhysRevE.97.042314}, and how the failure of nodes in one layer propagates to other layers \cite{PhysRevE.100.052306}.
It is essential to mention that the literature contains a wide variety of networks very similar in definition to multi-layered networks like multiplex networks \cite{Cardillo2013, 10.1145/2808797.2808852}, multilevel networks \cite{10.1145/3308558.3313484, 10.1145/1947940.1947942} and network of networks \cite{10.1145/2623330.2623643}.

Heterogeneous and multilayered network models are used separately for different problems. Further, multilayer network data structure had been referred as `heterogeneous multilayer' \cite{gyanendro-singh-etal-2020-sentiment, RYBALOVA2021110477, 8727718, Liu2024, 9756648, 10531187, Pio-Lopez2021} or `multilayered heterogeneous' \cite{9454288, 8683574} when used in many applications. The probable reason for the same is to express that the whole network is heterogeneous. However, the data structures used therein do not support heterogeneity in each layer, rather can be trivially reduced to the same definition used in \cite{10.1093/comnet/cnu016}. {A layered network with multiple nodes and edge types are considered heterogeneous in \cite{gyanendro-singh-etal-2020-sentiment, Liu2024, Pio-Lopez2021}. However, nodes of an individual type constitute a layer designating it as similar to multilayer network. The authors of \cite{9756648} consider each type of relation to be in a separate layer. Similarly, \cite{10531187}, considers the type of node (a person wearing a mask) to be fixed in a layer. In \cite{8683574}, the word heterogeneous is used to represent different non-overlapping communities in a layer. The nodes are homogeneous and, interestingly, some communities are shared between layers. This definition of layers and heterogeneity cannot be generalized to accommodate different network data.}

The literature review shows us that despite having different work on heterogeneous and multi-layered networks, there is no unified framework of network which can address heterogeneity and multi-layered property simultaneously. Furthermore, there is no data structure of network that can be used to capture various types of network, including homogeneous, heterogeneous and multi-layer network, depending upon the problem at hand.

\section{Hybrid Layered Network (HLN)}\label{definition}
In this section we will provide detail about the proposed Hybrid Layered Network (HLN). HLN is an unified data structure that can be heterogeneous and multi-layer simultaneously and thus able to represent complex graph-data. The data structure is flexible enough to express individual types of network as well as their complex combinations. In the latter part of this section we proved the generalization capabilities of the proposed data structure. 

\begin{figure}
     \centering
     \begin{subfigure}{0.6\textwidth}
         \centering
         \includegraphics[width=\textwidth]{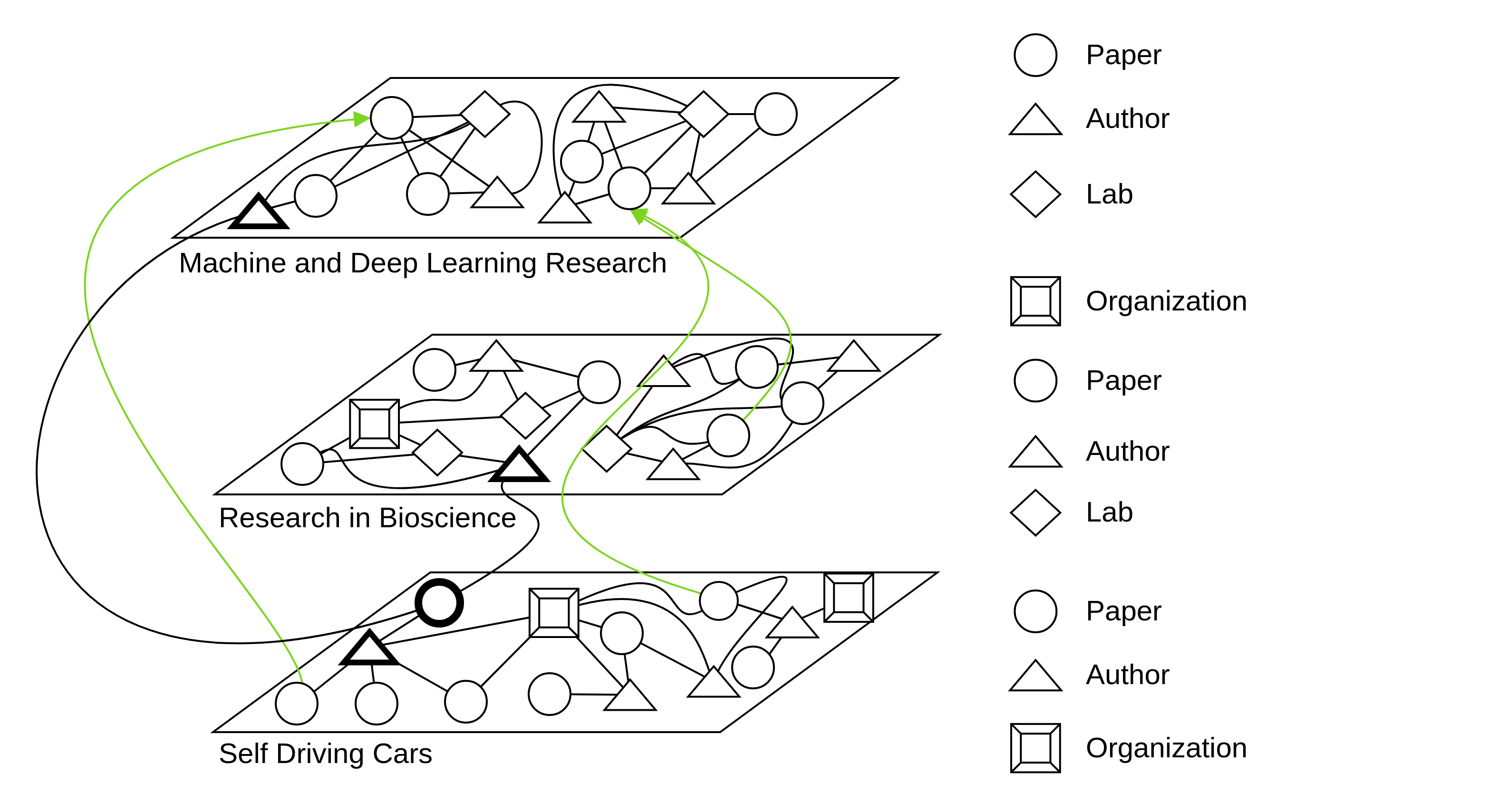}
         \caption{}
         \label{fig:fig2}
     \end{subfigure}
     \hfill
     \begin{subfigure}{0.3\textwidth}
         %\raisebox{10mm}
         \centering
         %\vspace*{10cm}
         \includegraphics[width=0.7\textwidth]{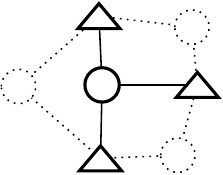}
         \caption{}
         \label{fig:fig2b}
     \end{subfigure}
     \caption{a) An example to demonstrate a hybrid-layered network \textcolor{black}{with node and edge types}. In the figure the black colored edges are undirected and the green colored edges are directed. b) The figure shows the author paper relation for the authors and papers marked in bold in Figure \ref{figure1label}(\subref{fig:fig2}). \textcolor{black}{The circle represents paper and the triangle represents author}. The dotted links and circles are the possibilities of two authors collaborating on a paper in the future.}
     \label{figure1label}
\end{figure}

\begin{definition}[Hybrid Layered Network]\label{defHLN}
A hybrid-layered network (HLN) is defined by quintuple $G = (V, E, L, T, \boldsymbol{\mathcal{R}})$ where $V$ is the set of vertices, $E \subseteq ((V \cross L) \cross (V \cross L))$ is the set of edges, $L$ is the set of layers, $T = \{T_V, T_E\}$ is the set of sets of vertex and edge types and $\boldsymbol{\mathcal{R}}$ is the set of functions. $\boldsymbol{\mathcal{R}}$ contains $3$ primary functions $R_{VT}$: \text{V $\rightarrow$ $T_V$}, $R_{ET}$: \text{E $\rightarrow$ $T_E$} to map vertex and edges to type and, $R_{VL}$: \text{V $\rightarrow$ $2^L \setminus \{\emptyset\}$} to map a vertex to a set of layers. 
\end{definition}

A vertex may be present in many layers and cannot exist outside layers, hence the function $R_{VL}$ maps a vertex to a power set of layers except the null set. For example a node $u$ belonging to $3$ layers $l_1, l_2, l_3$ will have $R_{VL}(v) = \{l_1, l_2, l_3\}$. We denote a node $v$ at layer $l$ as $v^l$ for sake of convenience, i.e., $l \in R_{VL}(v^l)$. There must be at least one layer in an HLN, i.e., $|L| \geq 1$. The set $T_V$ and $T_E$ at minimum contains one type \{$\bot$\} each.
An edge $e = (v_a^{l_b}, v_c^{l_d})$ denotes that there is a directed connection from $v_a$ at layer $l_b$ to $v_c$ at layer $l_d$. An edge is called an intra-edge if $l_b = l_d$ or inter-edge if $l_b \neq l_d$. The set $\boldsymbol{\mathcal{R}}$ contain functions for mapping nodes and edges to their respective types and layers. The functions $R_{VT}$ and $R_{ET}$ map a vertex and an edge to $T_V$ and $T_E$ respectively.

Let us see the model with an example shown in Figure \ref{fig:fig2}. The network contains three layers with layers representing research in the machine and deep learning (Layer 1), bio-science ( Layer 2), and self-driving cars (Layer 3). The figure contains three types of nodes in the Layers $1$ (paper, author, lab) and $2$ (paper, author, organization), and Layer $3$ contains four types of nodes namely paper, author, lab, and organization. There are directed interconnections between layers to show that papers in bio-science or self-driving cars cite another paper in machine and deep learning. The network cannot be represented with a homogeneous network without losing information. Even a multi-layered or heterogeneous network will fail to capture the network due to the presence of multiple types of nodes (and edges) and multiple layers respectively. It may be apparent that a heterogeneous network will describe Figure \ref{fig:fig2} by merging different layers into one. However, a heterogeneous network only partially captures the given network as described in Remark \ref{remarkHet}. 

\begin{rem}[Hybrid Layered Network is not another variant of a heterogeneous network]\label{remarkHet} The Figure \ref{fig:fig2} shows us in bold that three authors from three different layers have collaborated on a paper in the third layer. Consider that these three layers get merged into one, then these three authors and the paper will have the structure as shown in Figure \ref{fig:fig2b}. Once we avoid the layer structure all these three authors become homogeneous. Hence there is no way to separately keep their association with each other in terms of the subjectivity expressed by the layers. For example, the possibility of collaboration between the authors of layer $1$ with $3$, $2$ with $3$ and $1$ with $2$ can be different. The only way we can keep this subjectivity is through the layers while each layer is heterogeneous. HLN can preserve these constraints while the very definition of heterogeneous network is unable to capture the same. \end{rem}

\begin{lemma}\label{multilayerLemma}
The set of all multi-layered networks $\mathcal{M}$ is a subset of the set of all hybrid layered networks $\mathcal{H}_m$.
\end{lemma}
\begin{proof} We will prove this by contradiction. Let us assume that $\exists \ x = (Y, G_{intra}, G_{inter})$ such that $x \in  \mathcal{M}$ but $x \notin \mathcal{H}_m$. Now, $\exists \ y = (V, E, L, T, \boldsymbol{\mathcal{R}}) \in \mathcal{H}_m$ such that
\begin{align*}
    & L &= & \text{ } Y \\
    & G(V, E) &= & \text{ } G_{intra} \cup G_{inter} \text{ where}\\
    & G_{intra} &= & \text{ }\{G_1, G_2, \cdots, G_k\} \text{ where $G_i = (V_i, E_i)$} \\
    & G_{inter} &= &\text{ } \{G_{ij}\} \text{ where } G_{ij} = (V_i \cup V_j, E_{ij})\\
    & V_i &= &\text{ }\{\text{ }v \text{ }|\text{ } v \in V \text{ \& } i \in R_{VL}(v)\text{ }\} \\
    & E_i &= &\text{ }\{\text{ }(v_j, v_k) \text{ }| \text{ }(v_j^{l_i}, v_k^{l_i}) \in E\text{ }\} \\
    & E_{ij} &= &\text{ } \{\text{ }(v_k, v_m) \text{ }|\text{ } (v_k^{l_i}, v_m^{l_j}) \in E \text{ }\}\\
\end{align*}
The set $V = \bigcup_{i \in L} V_i$ and $E = \bigcup_{i \in L} E_i$.
Each of the graphs in a particular layer in $x$ is homogeneous (Definition \ref{defHomo}) but different layers may have different types of nodes. The functions $R_{VT}$ and $R_{ET}$ map all vertices and edges of a single layer to one value in the set $T_V$ and $T_E$ respectively. The presence of such a $y$ using which we can create an $x$ contradicts with our assumption. Thus we show that the set of all multi-layered networks is a subset of the set of all hybrid layered networks. \end{proof}

\begin{corollary}\label{multiplexcorollary}
A multiplex network is a special case of an HLN.
\end{corollary}
\begin{proof} A multiplex network is a multilayered network where every layer has the same vertex set so there is no need for interconnections between the layers. In a multiplex network $M = (Y, G_{intra}, G_{inter})$, $G_{intra}$ is $(G_1, G_2, \cdots,G_k)$ and $G_{inter} = \{G_{ij}\}$ where $G_i = (V_i, E_i)$ and $G_{ij} = (V_i \cup V_j, E_{ij})$ with $V_1= V_2 = V_3 = \cdots = V_{|Y|} = V$ and $E_{ij}= \emptyset$ $\forall G_{ij}$. In this case $T_V = \{T_1\}$ and $T_E = \{T'_{1},\cdots,T'_{|L|}\}$. Thus a multiplex network is a special case of a multi-layered network making it a special case of an HLN.
\end{proof}

\begin{lemma}\label{heterogeneoustoHLN}
The set of all heterogeneous networks $\mathcal{H}$ is a subset of the set of all hybrid layered networks $\mathcal{H}_m$.
\end{lemma}
\begin{proof} We will prove this by contradiction. Let us assume that $$\exists \ x = (V_{het}, E_{het}, \{A,B\},\{f_1, f_2\}) \in \mathcal{H}$$ such that $x \notin \mathcal{H}_m$. Now, $y = (V, E, L, T, \boldsymbol{\mathcal{R}})$ with the following values for the parameters is in $\mathcal{H}_m$.

\begin{tabular}{lllll}
& & & &\\
& & & $(V, E) = (V_{het}, E_{het})$  & $L = \{1\}$ \\
& & & $T_V = A$   & $T_E = B$  \\
& & & &
\end{tabular}

Note that $R_{VT} \equiv f_1$ and $R_{ET} \equiv f_2$ based on the definition of heterogeneous networks. The function $R_{VL}$ maps to default set as there is a single layer. Considering all the parameters of $y$ are generated from the parameters of $x$, it is proved that for every $x \in \mathcal{H}$ there exists a corresponding $y \in \mathcal{H}_m$ such that $x \equiv y$. Thus, $x \in \mathcal{H} \implies y \in \mathcal{H}_m$ where $x \equiv y$ which shows that $\mathcal{H} \subset \mathcal{H}_m$. We use the notation $\subset$ instead of $\subseteq$ as $\mathcal{H}$ can never be equal to $\mathcal{H}_m$ due to the presence of layers in $\mathcal{H}_m$.  \end{proof}

\begin{lemma}\label{homogeneoustoHLN}
The set of all homogeneous networks $\mathcal{S}$ is a subset of the set of all hybrid layered networks $\mathcal{H}_m$.
\end{lemma}
\begin{proof} We will prove this by contradiction. Let us assume that $\exists \ x = (V_{homo}, E_{homo})  \in \mathcal{S}$ such that $x \notin \mathcal{H}_m$. Now, $y = (V, E, L, T, \boldsymbol{\mathcal{R}})$ with the following values for the parameters is in $\mathcal{H}_m$.

\begin{tabular}{lllll}
& & & &\\
& & & $(V, E) = (V_{homo}, E_{homo})$ & $L = \{1\}$  \\
& & & $T_V =\{\bot\}$ & $T_E = \{\bot\}$ \\
& & & &
\end{tabular}

Note that $R_{VT}$ maps to $T_V$ and $R_{ET}$ maps to $T_E$. The function $R_{VL}$ maps to default set as there is a single layer. Considering all the parameters of $y$ are generated from the parameters of $x$, it is proved that for every $x \in \mathcal{S}$ there exists a corresponding $y \in \mathcal{H}_m$ such that $x \equiv y$. Thus, $x \in \mathcal{S} \implies y \in \mathcal{H}_m$ where $x \equiv y$ which shows that $\mathcal{S} \subset \mathcal{H}_m$.\end{proof}
\textcolor{black}{It must be noted that the definition of $\boldsymbol{{\mathcal{R}}}$ can contain additional functions. For example, all the nodes in a layer $L$ can be returned by a function say $R_{LV}$, and all nodes of type $T$ in a layer $L$ can be returned by a function say $R_{TL}$. When we want all node types in a layer we can represent $R_{TL}$ as $R_{L}$. In other words, we can add other functions to  $\boldsymbol{\mathcal{R}}$ as required. This makes the definition of HLN extendable for different contexts. The addition of the functions $R_{L}$ and $R_{TL}$ do not alter the definitions and proofs as mentioned earlier.}

\subsection{\textcolor{black}{Advantages of HLN through Examples}}
\label{sec:adv:hln}
\begin{figure}
  \centering
  \includegraphics[width=8cm,height=3.5cm]{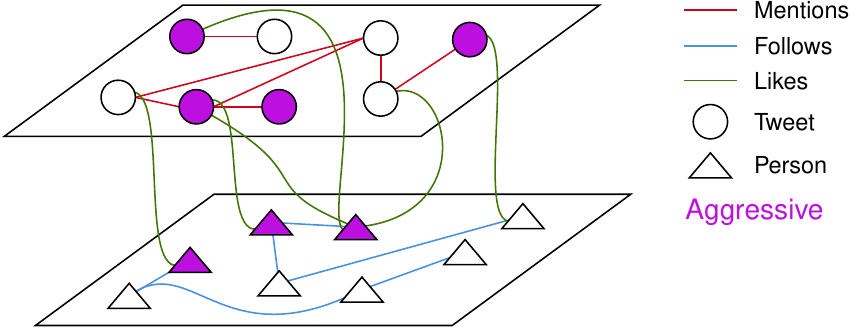}
  \caption{An example twitter network with each rectangle representing a layer. The first layer (represented by circles) contains tweets and the second layer (represented by triangles) contains users. }
  \label{twitter}
 \end{figure}
\textcolor{black}{Let us consider the network shown in the Figure \ref{twitter}. The network represents a real life twitter network which is heterogeneous and multi-layered at the same time. The first layer contains tweets {(represented by circles)}. There is a connection between two tweets using the same hashtag. The second layer (represented by triangles) contains users with links between two users indicating one follows the other. In addition to that a user or tweet can be aggressive or non aggressive represented using the color. The inter-layer links represent a user liking a tweet. The above network is represented as a HLN in Figure \ref{twitter}. We cannot represent this information using the definitions of \cite{gyanendro-singh-etal-2020-sentiment, Bianconi2018, 9756648, 10531187} as there are multiple types of relations in a layer and the same node is not be present in all the layers. One may argue, that the same information can be represented as a heterogeneous network \cite{7536145, 8683574} as shown in Figure \ref{twitter_hetero}, however, the complexity will increase in many folds as described here. In Figure \ref{twitter_hetero} we can see that the types of nodes doubled, $i.e.$, a node can be of two types (user, tweet) where each type can have two subtypes (aggressive, non-aggressive). Now let us consider a case where a user can be mildly aggressive, moderately aggressive, severely aggressive and non-aggressive. In that case a heterogeneous network will have $2*4 = 8$ types of nodes and $3$ types of edges making a total of $11$ types to substitute $2$ layers with $4$ types in a HLN. It must be noted that two layers can intrinsically mean $3$ types of edges ($2$ intra and $1$ inter) without explicit markup. Similarly, increasing one layer will require $3*4=12$ node types and $3(intra) + 3(inter) = 6$ edge types in a heterogeneous network increasing the total number of types to $18$ from $11$. In contrast, a HLN will require $3$ layers with $4$ types to represent the same. In fact, a HLN intrinsically stores $\Comb{|L|}{2}$ edge types which require explicit definitions in a heterogeneous network. This in-turn increases the computational complexity of certain tasks (example in the following paragraph) in a heterogeneous representation of the network. In other words, the proposed HLN model's advantage is at the abstraction level, which simplifies the data structure for complex graphs.} {When we compare our representation with existing representations like \cite{Pio-Lopez2021} we find that our model can store as many representations as necessary in a layer whereas existing models like MultiVERSE requires one layer for each type of node.}
\begin{figure}
  \centering
  \includegraphics[width=8cm,height=3.5cm]{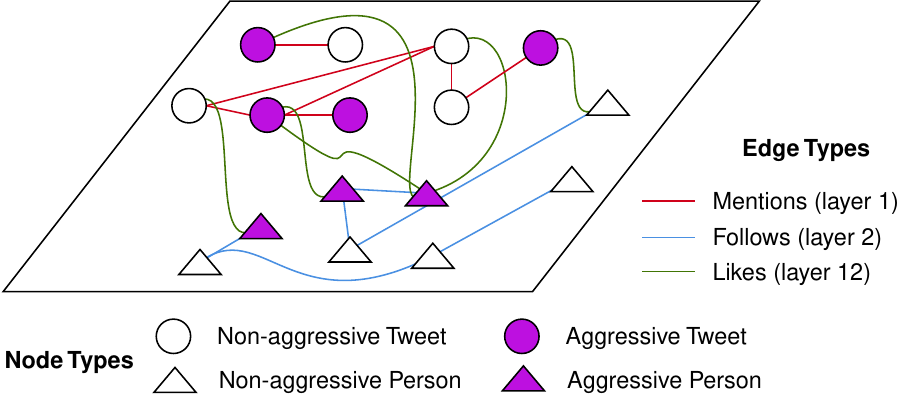}
  \caption{{An example showing twitter represented as a heterogeneous network.}}
  \label{twitter_hetero}
 \end{figure}

\textcolor{black}{Let us consider the application of HLN through the lens of link prediction in the given Twitter network. We consider the Jaccard Co-efficient for scoring a possible edge $(x,y)$ which can be defined as $JC(x,y) = \frac{N(x)\cap N(y)}{N(x) \cup N(y)}$ where $N(x)$ defines the neighbors of node $x$. In a heterogeneous setting we may need to consider the neighbors of a particular type. Considering the network in Figure \ref{twitter_hetero} if we need to find neighbors $N(x)$ of a particular node $x$ of type (say tweet) we need $O(V)$ time for each node $x$ in the worst case as shown in snippet $1$ below. In the same setting, we need $O(k)$ time in our HLN with the help of function $R_{TL}$ which returns all $k$ vertices of a particular type in a given layer, as shown in the snippet $2$. In a real Twitter network, $k << V$. We use the function $R_{VL}$ to obtain the layer information for node $x$.}
%In a homogeneous network with type information as the property of a node we will need $O(d = |N(x)|)$ in the worst case for finding neighbors of a particular type as shown in Algo. In a heterogeneous network with a specifically designed function to return type information inbuilt in the definition of the network we many need $O(n)$ in the worst case. In a HLN with layer information we can reduce the searching for type to searching within layers containing that particular type. The number of nodes in a layer $(l)$ in HLN $<<$ the number of nodes in the network reducing the complexity to $O(l)$.
\begin{Verbatim}[commandchars=\\\{\},codes={\catcode`$=3\catcode`_=8}]
1. for n in N(x):
    if n.type == x.type:
      N\textsubscript{T}(x).add(n)
2. for n in {N(x) U $R\textsubscript{TL}(x.type, R\textsubscript{VL}(x))$}:
    N\textsubscript{T}(x).add(n)
\end{Verbatim}
\subsubsection{Using HLN improves existing tasks}
The use of layers helps reduce types and query complexity as explained above. Now we show that layers can add additional  domain knowledge to an existing heterogeneous network. In fact, we can say that a layer is meta-information about a heterogeneous network that is inherently present but not apparently visible. HLN provides a systematic way of managing this meta-information starting from the definitions to the data structure.  {We use the link prediction task with and without layer information on the MovieLens \cite{10.1145/2827872} dataset to prove our claim. MovieLens is a heterogeneous dataset with two types of nodes namely movie and user with on type of node between them namely user rates movie. There are $9742$ user nodes, $610$ movie nodes and $100836$ edges. In the link prediction task we divide the edges in the dataset into two parts training, validation and testing. We train $2$ layers GNN models like \cite{10.5555/3294771.3294869,10.1609/aaai.v33i01.33014602, velickovic2018graph, bresson2018residual, brody2022how, ijcai2021p214, 10.5555/2969442.2969488, DBLP:journals/corr/abs-2006-07739, 10.5555/3495724.3497151, 10.5555/3524938.3525045} on the training dataset and try to predict the links on the validation and test dataset. We create two versions of the MovieLens dataset with and without layer numbers. In the version with layer numbers we encoded the domain knowledge that people who like at least one sci-fi movie rate movies differently when compared to people who like other genres of movies. We have moved sci-fi movies and users who have seen and rated at least one of them to layer $1$ and other users and movies to layer $2$. We encode the layer information into the feature vector by adding a bit where 0 represents layer 1 and 1 represents layer 2. In the unlayered version we use the actual feature vector which indicates all nodes belong to a single layer.} We try to predict links on such a network to see what movie a user will watch (and rate) next. For the link prediction task, we use state-of-the-art GNN architectures available. We run the GNN models on two different feature sets of the same network, the first being the feature set of the original heterogeneous network and the second with a layer dimension added to the feature vector of each movie and user. The results in Table \ref{mrTable4} show an increase in the area under the curve for all models when we use layers. {We use t-SNE (with perplexity = $30.0$, learning rate of $86.26$ and other parameter set to default as in the sklearn library) for reducing the dimensionality of the node features generated by our model to 2-dimensions}. The t-SNE plot shown in Figure \ref{embedding} shows that the embedding of movies belonging to the sci-fi genre is closer as is the embedding for the users watching that genre clearly showing the relevance of adding layers.
\begin{figure*}
     \centering
     \begin{subfigure}[b]{0.24\textwidth}
         \centering
         \includegraphics[width=\textwidth]{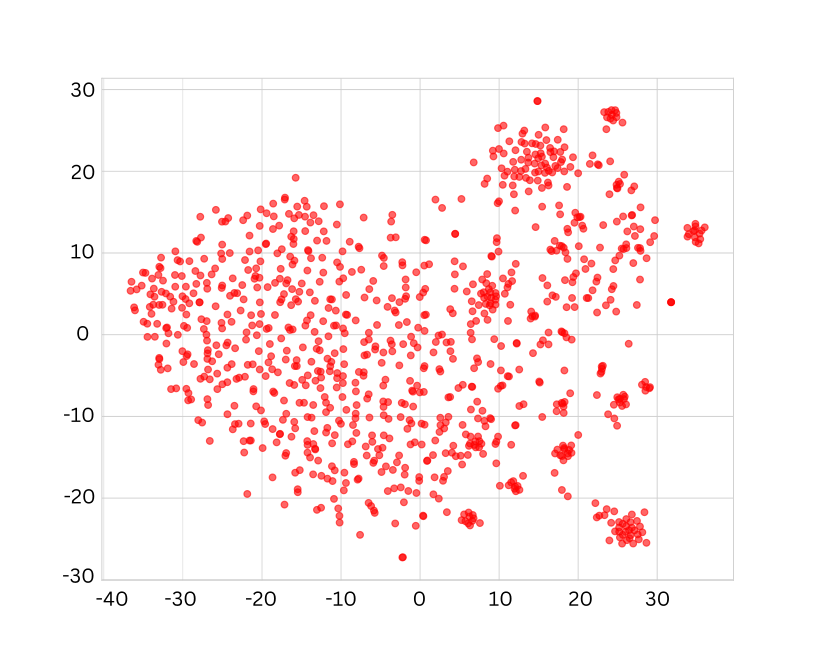}
         \caption{Movie - layer}
         %\label{fig:y equals x}
     \end{subfigure}
     %\hfill
     \begin{subfigure}[b]{0.24\textwidth}
         \centering
         \includegraphics[width=\textwidth]{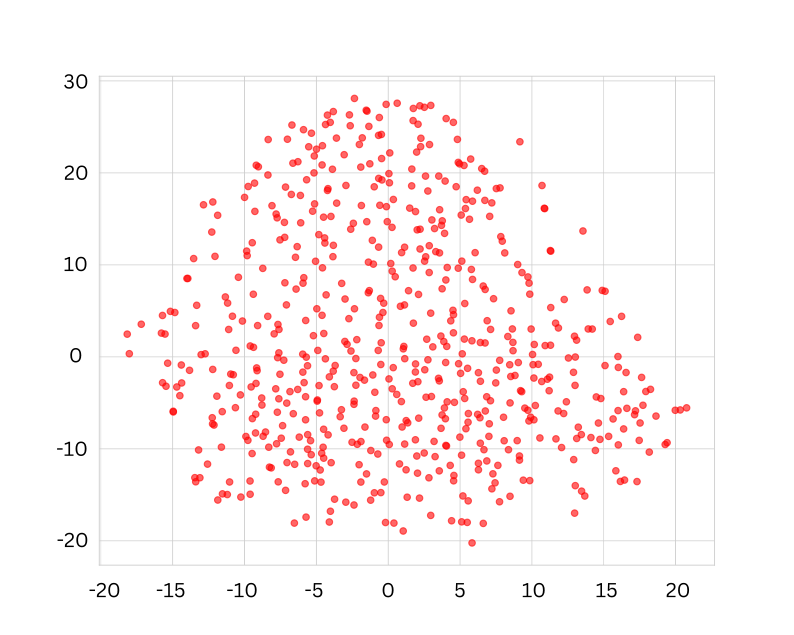}
         \caption{User - layer}
         %\label{fig:three sin x}
     \end{subfigure}
     %\hfill
     \begin{subfigure}[b]{0.24\textwidth}
         \centering
         \includegraphics[width=\textwidth]{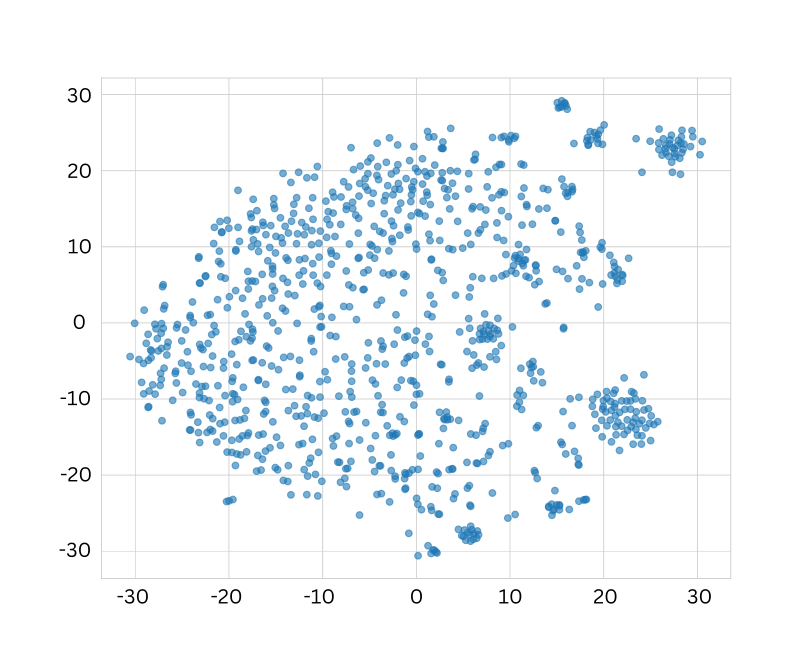}
         \caption{Movie + layer}
         %\label{fig:three sin x}
     \end{subfigure}
     %\hfill
     \begin{subfigure}[b]{0.24\textwidth}
         \centering
         \includegraphics[width=\textwidth]{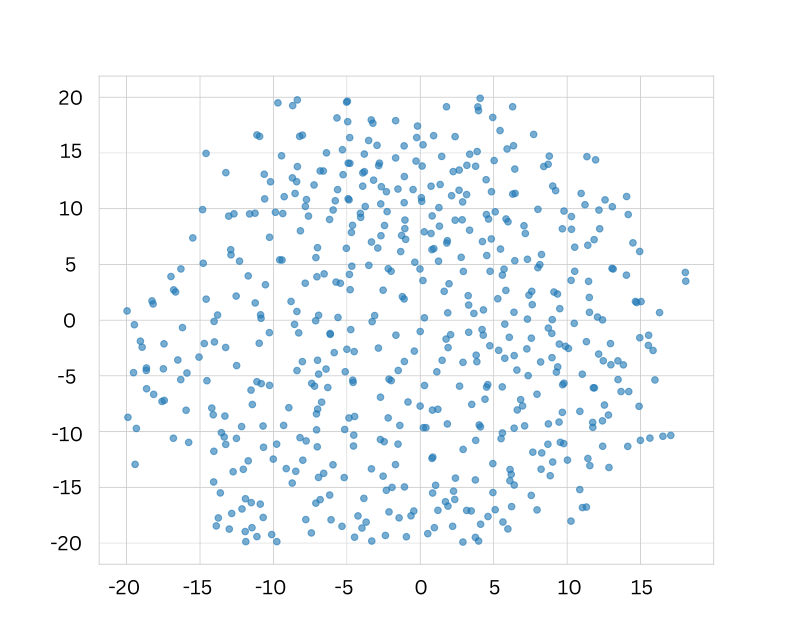}
         \caption{User + layer}
         %\label{fig:three sin x}
     \end{subfigure}
     \caption{Figures \textbf{a} and \textbf{b} show the t-SNE plot of the feature vector for the movie and user nodes without layer information. Figures \textbf{c} and \textbf{d} show the feature vectors after adding layer information. We can clearly see that the feature vectors have condensed after adding layer information bringing structurally similar nodes closer to one another thus increasing the result of link prediction. Features are obtained from GraphSAGE model.}
     \label{embedding}
\end{figure*}
\begin{table}[h]
    \centering
    \caption{\textbf{AUC} for the Link Prediction Task with and without layers on the movie-lens heterogeneous dataset.}
    
    \begin{tabular}{|l|l|l|}
    \hline
        \textbf{Model} & \textbf{Without Layer} & \textbf{With Layer} \\ \hline
        \rule{0pt}{2ex}%  EXTRA vertical height
        %~ & \textbf{AUC} & \textbf{AUC} \\ \hline
        \rule{0pt}{2ex}%  EXTRA vertical height
        SAGEConv \cite{10.5555/3294771.3294869} & 92.77 & \textbf{94.11} \\ \hline
        \rule{0pt}{2ex}%  EXTRA vertical height
        Graphconv \cite{10.1609/aaai.v33i01.33014602} & 76.2 & \textbf{80.87} \\ \hline
        \rule{0pt}{2ex}%  EXTRA vertical height
        GATConv \cite{velickovic2018graph} & 88.75 & \textbf{89.77} \\ \hline
        \rule{0pt}{2ex}%  EXTRA vertical height
        ResGatedGraphConv \cite{bresson2018residual} & 86.62 & 86.63 \\ \hline
        \rule{0pt}{2ex}%  EXTRA vertical height
        GATv2Conv \cite{brody2022how} & 89.56 & \textbf{90.78} \\ \hline
        \rule{0pt}{2ex}%  EXTRA vertical height
        TransformerConv \cite{ijcai2021p214} & 92.87 & \textbf{94.13} \\ \hline
        \rule{0pt}{2ex}%  EXTRA vertical height
        MFConv \cite{10.5555/2969442.2969488} & 81.89 & \textbf{85.46} \\ \hline
        \rule{0pt}{2ex}%  EXTRA vertical height
        GENConv \cite{DBLP:journals/corr/abs-2006-07739} & 93.43 & \textbf{94.51} \\ \hline
        \rule{0pt}{2ex}%  EXTRA vertical height
        GeneralConv \cite{10.5555/3495724.3497151} & 77.61 & \textbf{83.84} \\ \hline
        \rule{0pt}{2ex}%  EXTRA vertical height
        FiLMConv \cite{10.5555/3524938.3525045} & 75.07 & \textbf{80.34} \\ \hline
    \end{tabular}
    
    \label{mrTable4}
\end{table}
\section{Synthetic HLN Generation}
\label{algorithm}
We proposed a novel parameterized algorithm for generating different hybrid layered networks in this section. The proposed algorithm can generate a multi-layered, heterogeneous, and homogeneous network using different values of the parameters as described in the Lemmas \ref{multilayerLemma}, \ref{heterogeneoustoHLN}, and \ref{homogeneoustoHLN} respectively.

\subsection{Algorithm}
The Algorithms of \ref{algo1} and \ref{algo2} generate an HLN. The algorithms work as follows. At time step $t$ a new node $n$ is added to the initially empty network $G$. The number of nodes that can be added to the network is limited by the parameter $N$. The node $u$ is assigned to a layer $l$ of $L$ uniformly randomly. The type of $u$ is assigned uniformly randomly from $R_L(l)$ where $R_L(i) = \{T_i\}_{i \in L}$ and $T_i \subseteq T_V$. \textcolor{black}{Since different real-world datasets have different type distributions, we can assign node types based on any other distribution without changing any other part of the algorithm}. The added node $u$ connects with other nodes in the same layer and other layers using \textcolor{black}{preferential attachment. The preference of a node is decided based on its own degree and the degree of its neighbours (Algorithm \ref{algo2}, Lines 6-7). This makes the algorithm capable of generating power law and other types of networks.  The parameters $\alpha$ and $\beta$ decide the weightage to be given to a node's own degree and the degree of its neighbours, respectively.} The minimum number of connections a node makes with other nodes (in the same and different layers) is decided by the parameter $m$, a $L\times L$ matrix. For making intra-layer connections, the function \textit{connection1} takes an induced graph $G_{ii}$ from the HLN $G$ where $G_{ij} = I(G, V_i \cup V_j)$. The induced graph can also be called a subHLN. An induced graph comes with all the types of nodes and edges associated with vertex set $V_i$ in layer $i$ and $V_j$ in layer $j$. For the inter-layer connections with the node $u$, the function \textit{connection2} takes three induced graphs $G_{ii}$, $G_{jj}$ and $G_{ij}$ where $i = l$ and $j \in L \setminus l$. Let us consider a situation where no nodes are in the inter-layer subHLN $G_{12}$ and a new node ($u$) is assigned to the layer $1$. In making a cross-layer connection with layer $2$, if there are sufficient nodes ($> m_{12}$) in the layer $2$, then $m_{12}$ nodes are selected at random from layer $2$ and connected with $u$. If the number of nodes in layer $2$ $< m_{12}$, then we store the current node in a list so that the edges with $u$ can be created once there are sufficient nodes in layer $2$. 

\begin{rem}
The above algorithm assigns a node to a single layer, making it incapable of generating a multiplex network without certain modifications.
\end{rem}
\subsection{\textcolor{black}{Complexity Analysis and Scalability}}
\textcolor{black}{Adding a new node in layer $l$ triggers two functions, \textit{conn1} and \textit{conn2}, for intra layer and inter layer link generation. The algorithm establishes the intra layer links in $O(m)$ time where $m = M_{ll}$. Here $M_{ll}$ denotes the minimum number of connections a node makes in its layer. Following the intra layer links, we make inter layer connections for the node with all other layers in $O(|L|m')$ considering the worst case. where $|L|$ denotes the number of layers and $m' = \max\limits_{j \in L}{(M_{lj})}$. It must be noted that the worst case will arise when the node needs to connect to every other node in every layer.}

\DontPrintSemicolon
\begin{algorithm}[t]
\caption{Generating an HLN}
\label{algo1}
    \KwIn{$N, L, R_L, m, \textcolor{black}{\alpha, \beta}$}
    \KwOut{Return HLN G}
    $node \leftarrow 1, R_{VL} \leftarrow \{\}$\;
    $G = \hbox to 1.5cm{\textit{Empty HLN}}$\;
    \While {$node < N$}{
        $i = uniformRandom(L)$ \;
        %\tcc{\color{gray}${\textit{Pick uniform randomly from L}}$}\;
        $R_{VL}(node) = i$ \;
        %\tcc{\textit{Assign that vertex to that layer}}\;
        $t = uniformRandom(R_L(i))$\;
        {$G_{ii}.addNode(node)$}\;
        \textcolor{black}{\hbox to 5cm{$G_{ii} = \textit{connection1}(G_{ii}, node, M_{ii}, \alpha, \beta)$} \;}
        %\tcc{\color{gray}\textit{function to create intra/inter layer edges}}\;
        \While{$j \in L$}{
            \If {$i \neq j$}{
                %\tcc{\color{gray}$\textit{A graph with 0 nodes}$}
                \textcolor{black}{$G_{ij} = {\textit{connection2}(G_{ii}, G_{jj}, G_{ij}, node, M_{ij}, \alpha, \beta)}$ \;}
            }
        }
        $node = node + 1$;
    }
    %\SetKwInOut{Input}{Input}
    %\SetKwInOut{Output}{Output}
\end{algorithm}

\begin{algorithm}
\caption{connection1() for intra layer and connection2() for inter layer connections}
\label{algo2}
\KwIn{$g_1, g_2(optional), g_3(optional), node, m, \textcolor{black}{\alpha, \beta}$}
$G = emptyGraph()$\;
\textcolor{black}{\Fn{nodeDistribution ($G, \alpha, \beta$)}
{
    \textcolor{black}{{$nodesDist = [\text{ }]$}\;
        \For{$i \in G.nodes$}
        {
            {$j = 0$}\;
            {$c = int(\alpha * f_{deg}(i) + \beta * f_{neighbor\_deg}(i))$}\;
            \While{$j < c$}
            {
                {$nodesDist.add(i)$}\;
            }
        }}
}}
\Fn{connection1 ($g_1, node, m, \alpha, \beta$)}
{
    \If {$count(g_1.nodes) < m$ }
    {
        \textit{return $G = g_1$} \;
        %\textit{return $G$} \;
        %\tcc{\color{gray}\textit{$g_1$ does not have sufficient nodes for intra connections}}
    }
    \If { $count(g_1.nodes) = m$ }
    {
        {$G = starGraph(g_1.nodes)$} \;
    %\tcc{\color{gray}\textit{make a star graph out of the nodes of g1}}
        %\textit{return $G$}\;
    }
    \If{$g_1.nodes > m$}
    {
        {$G.edges = g_1.edges$}\;
        %{$repeated\_nodes = [i \textit{ }| \textit{ } i \in G.(i,j) \textit{ or } i \in G.(j,i)]$}\; 
        \textcolor{black}{{$nodesDist = nodeDistribution(G, \alpha, \beta$})\;
        % \textcolor{black}{{$repeated\_nodes = [\text{ }]$}\;
        % \For{$i \in G.nodes$}
        % {
        %     {$j = 0$}\;
        %     {$c = int(\alpha * f_{deg}(i) + \beta * f_{neighbor\_deg}(i))$}\;
        %     \While{$j < c$}
        %     {
        %         {$repeated\_nodes.add(i)$}\;
        %     }
        % }}
        %\textcolor{black}{{$repeated\_nodes = [i \textit{ }| \textit{ } for \textit{ } i \in G.nodes for\textit{ } j \in}}\alpha*f_{deg}(i) + \beta*f_{neighbor-deg}(i) $}}\;
        %\tcc{\color{gray}\textit{Contains repeated items}}
        {$targets = uniformRandom(nodesDist, m)$}\;}
        %\tcc{\color{gray}\textit{select m unique nodes from repeatedNodes}}
        {$newEdges = [(node,i)\textit{ }|\textit{ }i \in targets]$}\;
        %\tcc{\color{gray}\textit{Add m edges from 'node' to item in}}
        {$G.addEdges(newEdges)$} \;
        %\textit{return $G$}\;
    }
    \KwOut{return intra graph $G$}
}
\Fn{connection2 ($g_1, g_2, g_3, node, m, \alpha, \beta$)}
{
    \If{ $count(g_3.edges) = 0$ and $count(g_2.nodes) < m$}
    {
        {$g_3.add(node)$} \;
        %\tcp{\color{gray}\textit{There are insufficient nodes in the graph $g_2$ to connect to Add the current node in $g_3$ for future edge creation}}
        \textit{return $g_3$} \;
    }
    
    \If{$count(g_3.edges) = 0$}
    {
        $targets = uniformRandom(g_2.nodes, m)$ \;
        %\tcc{\color{gray}\textit{reaching this condition implies that $count(g_3.nodes) > 0$}}
        \For{$vertex \in g_3.nodes$}
        {
            \For{each \textbf{item} in $targets$}
            {
                $g_3.addEdge((vertex,item))$ \;
            }
        }
    }
    $a = g_2.edges$ \;
    $b = g_3.edges$ \;
    $a = a.add(b)$ \;
    $G.addEdges(a)$ \;
    %\tcc{\color{gray}\textit{add all the edges into a graph G}}
    %$G.addNodes(g_2.nodes)$ \;
    %\tcc{\color{gray}\textit{Handles the case where m for $g_{ij}$ is $< m$ for $g_j$}}
    \textcolor{black}{$nodesDist = nodesDistribution(G, \alpha, \beta)$ \;% \textit{ or } i \in g_2.nodes]$ \;
    %$repeatedNodes = [i \textit{ }| \textit{ }i \in G.(i,j) \textit{ or } G.(j,i)]$ \;% \textit{ or } i \in g_2.nodes]$ \;
    %$repeatedNodes=[i\textit{ }|\textit{ } i \not \in g_1.nodes]$ \;
    $nodesDist=[i\textit{ }|\textit{ } i \in nodesDist \text{ $\&$ } i \not \in g_1.nodes]$ \;
    $targets = uniformRandom(nodesDist, m)$ \;}
    %\tcc{\color{gray}\textit{select m unique nodes}}
    \For{each \textbf{item} in $targets$}
    {
        $g_3.addEdge((node,item))$ \;
    }
    %\textit{return $g_3$} \;
    \KwOut{return inter graph $g_3$}
}
%\KwOut{Return intra/inter graph}
\end{algorithm}

%----------------------------------------------Results--------------------------------------------------

\section{Structural Measures of HLN} \label{neighborhood}
In this section, we define some of the structural measures of HLN. By definition HLN uses direction for an edge $e \in E$. However, many network measures consider the in-links and out-links together. When applicable, notations related to in-links and out-links are superscripted with $IN$ and $OUT$ respectively in the following text.

\begin{definition}[Out/In-Neighborhood in HLN]
The Out/In-neighborhood of a node $v^l$ is defined by the connected nodes from/to the node $v^l$ to all the nodes situated in any layer in a set of layers $\boldsymbol{\mathcal{L}}$ and having a type $t \in \boldsymbol{\mathcal{T}}$. That is, 

\begin{equation}
\label{eqn4}
\footnotesize
\begin{aligned}
    N^{IN}(v^l, \mathcal{L}, \mathcal{T}) & = & \{u^k |(u^{k}, v^l)\in E, k \in \mathcal{L}, R_{VT}(u) \in \mathcal{T} \} \\
    N^{OUT}(v^l, \mathcal{L}, \mathcal{T}) & = & \{u^k |(v^l, u^{k})\in E, k \in \mathcal{L}, R_{VT}(u) \in \mathcal{T} \}
 \end{aligned}
 \end{equation}
\end{definition}
\begin{rem}
The definition of neighborhood is flexible to include as many types of nodes and layers we wish to take. To get all types of neighbors in all the layers we set $\boldsymbol{\mathcal{T}} = T_V$ and $\boldsymbol{\mathcal{L}} = L$ where $T_V$ and $L$ denote all vertices and layers respectively.
\end{rem}
\begin{rem}
 A node can be present in more than one layers. One should note that the definition of neighborhood presented here does not contain the neighbors that the same node $v$ in layer $k$ may have where $k \neq l$.
\end{rem}

 \begin{definition}[Neighborhood in HLN]\label{defNeighborhoodHLN}
The neighborhood of a node $v^l$ is defined by $ N(v^l, \mathcal{L}, \mathcal{T}) = N^{IN}(v^l, \mathcal{L}, \mathcal{T}) \cup N^{OUT}(v^l, \mathcal{L}, \mathcal{T})$.
\end{definition}

\begin{definition}[HLN Degree Centrality]\label{HLNDeg}
Given $\mathcal{L}$ and $\mathcal{T}$ the degree centrality ($DC$) of a node $v^l$ in an HLN is the ratio of the number of neighboring nodes of $v^l$ having type in $\mathcal{T}$ and belonging to a layer in $\mathcal{L}$ to the count of all nodes of type in $\mathcal{T}$ and any layer in $\mathcal{L}$. That is,
\end{definition}
\begin{equation}
\footnotesize
\label{eqn9}
\begin{gathered}
DC(v^l, \mathcal{L}, \mathcal{T}) = \frac{|N(v^l, \mathcal{L}, \mathcal{T})|}{|\{u^k | k \in \mathcal{L}, u \in V, u^k \neq v^l, R_{VT}(u) \in \mathcal{T} \}|}
\end{gathered}
\end{equation}
\begin{definition}[Shortest Path in HLN]
Given $\mathcal{L}$ and $\mathcal{T}$ the shortest path between two nodes $v^l$ and $w^k$ in an HLN is a path through the nodes of any layer in $\mathcal{L}$ and type in $\mathcal{T}$ such that the sum of the weights (in case of a unweighted HLN the weights of all edges are $1$) of the edges in the path is minimized. There can be more than one shortest path between two nodes and the set of all such shortest paths is denoted by $sp(v^l,w^k)$. The quantity $d(v^l, w^k)$ is the sum of the weights on the edges of a shortest path between $v^l$ and $w^k$. When there is no path between $v^l$ and $w^k$ the $d(v^l, w^k)$ is $\infty$.
\end{definition}

\begin{definition}[HLN Betweeness Centrality]
Given $\mathcal{L}$ and $\mathcal{T}$ the betweeness centrality of a node $v^l$ in a hybrid layered network is the fraction of shortest paths between any two nodes $x^k$ and $y^j$ (where $R_{VT}(x),R_{VT}(y)\in\mathcal{T}, k,j\in\mathcal{L}$) passing through node $v^l$ among all the shortest paths between $x^k$ and $y^j$. If there is no path between $x^k$ and $y^j$ then $\frac{|sp(x^k,y^j|v^l)|}{|sp(x^k,y^j)|}$ is considered to be $0$. That is,
 \end{definition}
 
\begin{equation}
\label{eqn10}
\footnotesize
\begin{gathered}
BC(v^l, \mathcal{L}, \mathcal{T}) = \sum_{x^k,y^j \in V'}{\frac{|sp(x^k,y^j|v^l)|}{|sp(x^k,y^j)|}}\\\text{ where }
V' = \{u^i | i \in \mathcal{L}, R_{VT}(u) \in \mathcal{T}, u \in V, u^i \neq v^l \}
\end{gathered}
\end{equation}
If we are considering only cross layered connections then we can set the $layers$ variable to $L - R_{VL}(v^l)$. The cross layered betweeness will indicate the importance of a node outside its own layer.
 \begin{definition}[HLN Closeness Centrality]
Given $\mathcal{L}$ and $\mathcal{T}$ the closeness centrality of a node $v^l$ in a hybrid layered network is the average shortest path length from $v^l$ to all other nodes of a layer in $\mathcal{L}$ and type in $\mathcal{T}$ in the network. 
\end{definition}
\begin{equation}
\label{eqn12}
\footnotesize
\begin{gathered}
CC(v^l, \mathcal{L}, \mathcal{T}) = \sum_{u^k \in V'}\frac{1}{d(v^l, u^k)} \\
\text{ where } V' = \{ u^k | u \in V, u^k \neq v^l, k \in \mathcal{L}, R_{VT}(u) \in \mathcal{T} \}
\end{gathered}
\end{equation}

 \begin{lemma}\label{lemmahomogeneous}
 Given an HLN $G = (V, E, T, \mathcal{R})$ with $|L| = 1$ and $T_V, T_E = \{\bot\}$, i.e., when an HLN is a homogeneous network (Lemma \ref{homogeneoustoHLN}), the neighborhood of a node $v^l \in V$ is equivalent to the neighborhood of $v$ in a homogeneous network.
 \end{lemma}

\begin{proof} 
Given HLN $G$ is nothing but a homogeneous network as per Lemma \ref{homogeneoustoHLN}. The Definition \ref{defNeighborhoodHLN} considers all types of neighbors of a node $v^l$ in all the layers when $|L| = 1, T_V,T_E = \{\bot\}$ which is nothing but the degree of the node $v^1$; making the neighborhood of HLN equivalent to the neighborhood of the homogeneous network (Definition \ref{neighHomo}, defined below) it represents. \end{proof}

\begin{definition}[Neighborhood] \label{neighHomo}
The neighbourhood of a node $v$ in a homogeneous network are the nodes that have an edge with $v$ i.e. the neighbourhood of a node $v$ is as defined as,
\begin{equation}
\label{eqnneighHomo}
\begin{gathered}
    N(v) = \{v_j |\text{ } (v,v_j) \in E\}
 \end{gathered}
 \end{equation}
\end{definition}

\begin{corollary}\label{degHomogeneous}
Given an HLN which is a homogeneous network (Lemma \ref{lemmahomogeneous}), $DC(v^l, \mathcal{L}, \mathcal{T})\\\text{ (Equation }\ref{eqn9}\text{) }\equiv DC(v)\ref{eqndegHomo}\text{) }$.
\begin{equation}
\label{eqndegHomo}
\begin{gathered}
   DC(v) = \frac{1}{n-1}|N(v)|
\end{gathered}
\end{equation}
\end{corollary}

\begin{proof} 
The neighborhood of an HLN with parameters according to Lemma \ref{lemmahomogeneous} is equivalent to the neighborhood of a homogeneous network. Thus, the numerator in Equation \ref{eqn9} is equivalent to the number of neighbors of a node (making numerator in Equation \ref{eqn9} = Equation \ref{eqndegHomo}). The denominator in Equation \ref{eqn9} contains all the nodes of the network (an HLN equivalent to a homogeneous network) except $v^l$ (the node whose centrality we are trying to find). So, the denominator in Equation \ref{eqn9} is equivalent to the denominator in Equation \ref{eqndegHomo}. Thus, it is proved that $DC(v^l, \mathcal{L}, \mathcal{T}) \equiv DC(v)$.
\end{proof}
\begin{corollary}\label{coroSPHomo}
The shortest path between two nodes of an HLN with parameters $|L| = 1, T_V, T_E = \{\bot\}$ is equivalent to the shortest path between the same nodes in a homogeneous network.
\end{corollary}
\begin{proof}
When we have only a single layer, i.e., $|L| = 1$ and a single type of vertex and edge , i.e., $T_V, T_E = \{\bot\}$ then we consider nodes belonging to all the layers and node types in the shortest path by default (as an HLN is a homogeneous network with the given parameters as per Lemma \ref{homogeneoustoHLN}) making the shortest path in an HLN equivalent to the shortest path in a homogeneous network.
\end{proof}
\begin{corollary}\label{betweenHomogeneous}
Given an HLN which is a homogeneous network (Lemma \ref{lemmahomogeneous}), $BC(v^l, \mathcal{L}, \mathcal{T})\\\text{ (Equation }\ref{eqn10}\text{) }\ref{eqnbetHomo}\text{) }$.
\begin{equation}
\label{eqnbetHomo}
\begin{gathered}
BC(v) = \sum_{x,y \in V\setminus \{v\}}{\frac{|sp(x,y|v)|}{|sp(x,y)|}}
\end{gathered}
\end{equation}
The function $sp(x,y)$ denotes the set of all shortest paths between two nodes $x$ and $y$ in a network and $sp(x,y|v)$ returns the shortest paths from $x$ to $y$ that passes through $v$.
\end{corollary}
\begin{proof} 
The shortest path between two nodes of an HLN with parameters as in Corollary \ref{coroSPHomo} is equivalent to the shortest path between the same nodes of a homogeneous network. Thus the numerator and denominator in Equation \ref{eqn10} is equivalent to the numerator and denominator in Equation \ref{eqnbetHomo}. 
\end{proof}

\begin{corollary}
Given an HLN which is a homogeneous network (Lemma \ref{lemmahomogeneous}), $CC(v^l, \mathcal{L}, \mathcal{T})$ \\\text{ (Equation }\ref{eqn12}\text{) } $\equiv CC(v)\ref{eqncloseHomo}\text{) }$.
\begin{equation}
\label{eqncloseHomo}
    \begin{aligned}
        \begin{gathered}
            CC(v) = \sum_{w \in V \setminus \{v\}} \frac{1}{distance(v,w)}
        \end{gathered}
    \end{aligned}
\end{equation}
Here (\textit{distance}($v$,$w$)) denotes the sum of the weights on the edges of the shortest path between two nodes $v$ and $w$ in a network.
\end{corollary}
\begin{proof}
The Distance between two nodes of an HLN with parameters as in Corollary \ref{coroSPHomo} is equivalent to the distance between the same nodes of a homogeneous network. Thus the numerator and denominator in Equation \ref{eqn12} is equivalent to the numerator and denominator in Equation \ref{eqncloseHomo}. 
\end{proof}
\textcolor{black}{
\begin{definition}[HLN Clustering Co-efficient]
Given $\mathcal{L}$ and $\mathcal{T}$, the clustering coefficient (CCo) of a node, $v^l$, in a hybrid layered network is defined as the fraction of triangles that the node $v^l$ participates in, out of the total number of triangles possible through that node. That is,
\begin{equation}
\label{eqnClusterC}
\footnotesize
\begin{gathered}
CCo (v^l, \mathcal{L}, \mathcal{T}) = \frac{2*|Triangles(x^k , y^j, v^l)|}{|N(v^l, \mathcal{L}, \mathcal{T})| * (|N(v^l, \mathcal{L}, \mathcal{T})| - 1)} \\
\text{ where } k, j \in \mathcal{L}, R_{VT}(x^k) \in \mathcal{T}, R_{VT}(y^j) \in \mathcal{T}
\end{gathered}
\end{equation}
%We consider only those triangles where the two nodes other than $v^l$ belong to $\mathcal{L}$ and $\mathcal{T}$.
We can prove all the Lemmas and Corollary for clustering co-efficient in a similar manner as shown in the previous definitions.
\end{definition}}

\section{Experiment and Results}
\label{results}
Experiments have been performed to show the ability of the proposed algorithm to generate different networks with structural properties close to real-world networks by changing the algorithm parameters. We use Twitter and European air transportation network as representative of HLN and multilayer network respectively. Additionally, we use chemical, biological and various other network as representative of homogeneous network. %Degree distributions and centrality measures are used for comparing the generated synthetic network with their real counterparts. 

\subsection{\textcolor{black}{Real world Hybrid Layered Network Generation}}
\textcolor{black}{We have used a Twitter dataset \cite{guptashubham} (referred to here as \textit{TWITT}) and represented it in the format as shown in Figure \ref{twitter}. The rationale for the use of TWITT is that the network is better modeled (refer Section \ref{sec:adv:hln}) as a HLN with heterogenity in its layers. A tweet can be easily classified as aggressive or non-aggressive based on a standard language classification model. The Twitter network has $20125$ nodes and $3938046$ edges in the tweet layer (layer $1$), $35936$ nodes and $60824$ edges in the user layer (layer $2$) and $93123$ edges in the user-tweet layer (interlayer connection).}
\subsubsection{\textcolor{black}{Comparative Results}}
\begin{figure*}
  \centering
    \begin{subfigure}[b]{0.31\textwidth}
          \centering
          \includegraphics[width=\textwidth]{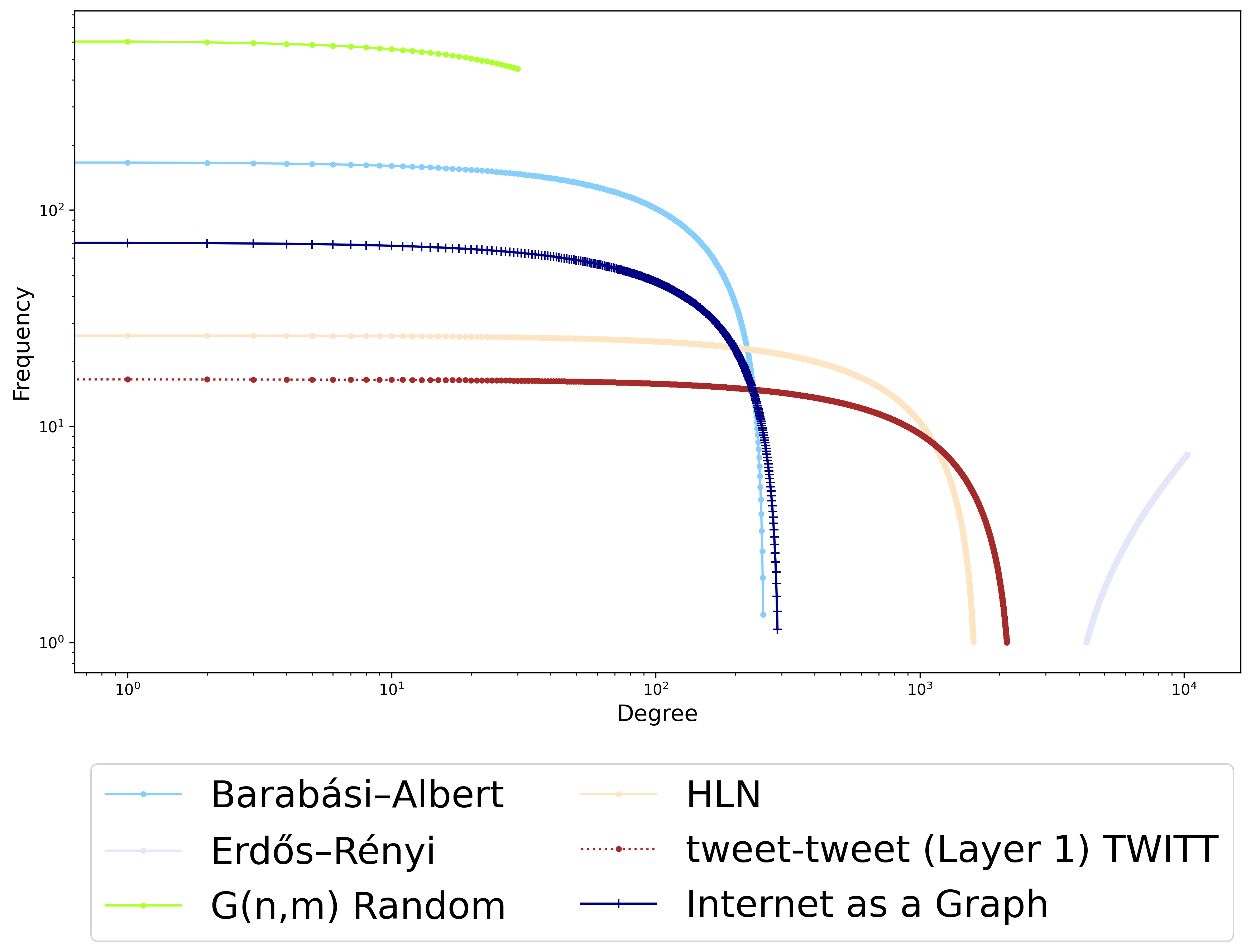}
          \caption{tweet $\xleftrightarrow{\text{tag}}$ tweet layer }
          \label{layer1_comparison}
    \end{subfigure}
    \begin{subfigure}[b]{0.31\textwidth}
          \centering
          \includegraphics[width=\textwidth]{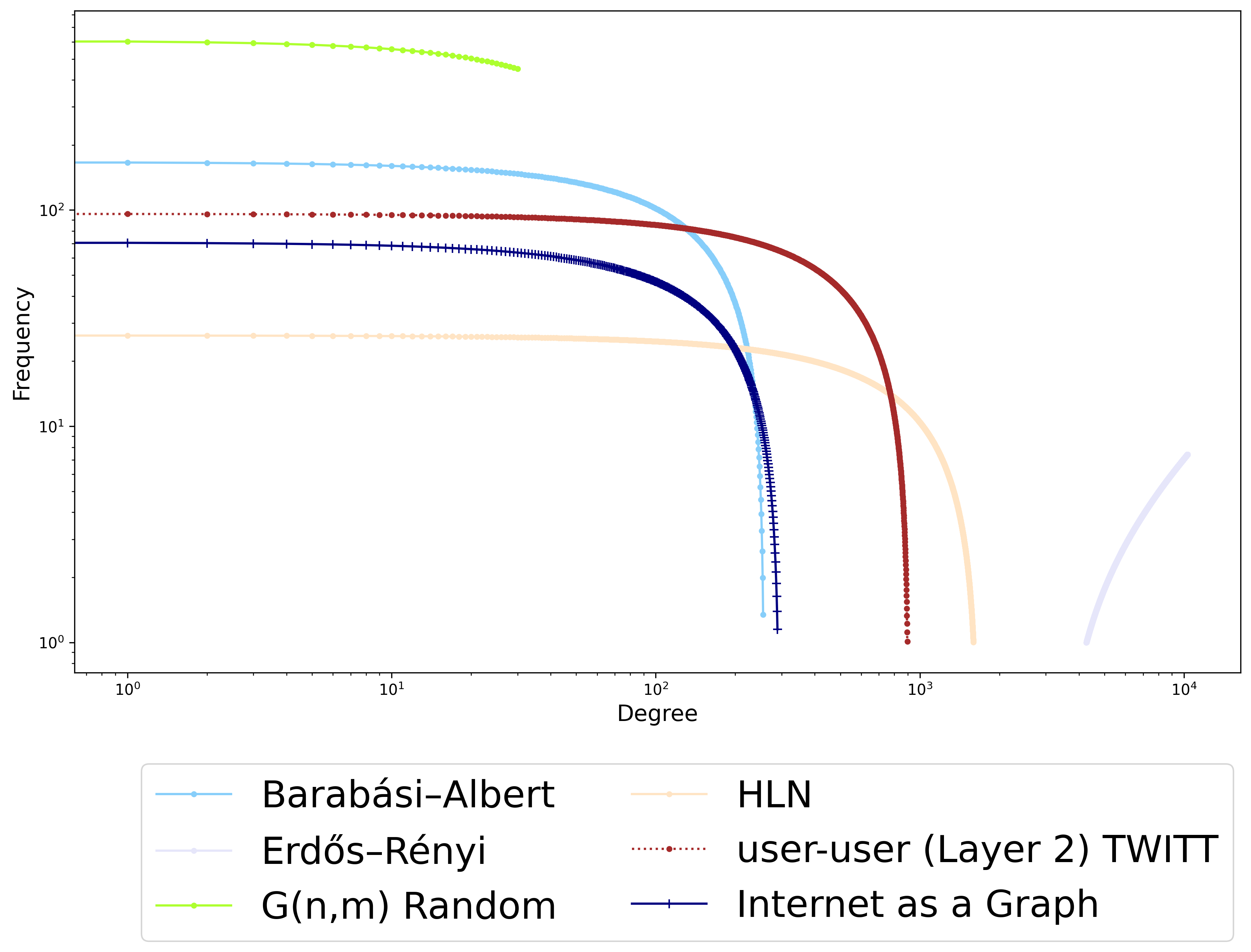}
          \caption{user $\xleftrightarrow{\text{follows}}$ user layer. }
          \label{layer2_comparison}
     \end{subfigure}
     \begin{subfigure}[b]{0.31\textwidth}
          \centering
          \includegraphics[width=\textwidth]{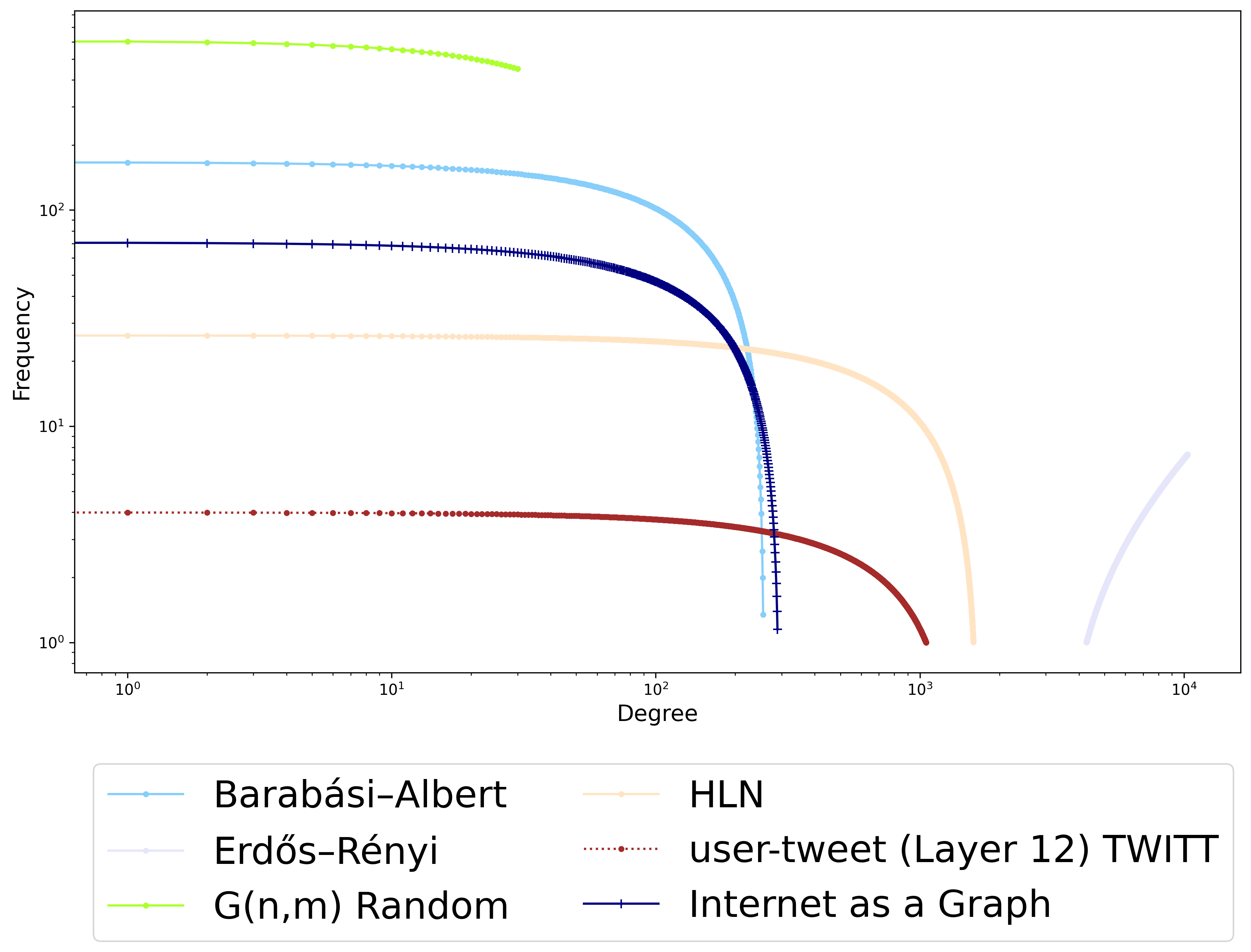}
          \caption{user $\xleftrightarrow{\text{posts}}$ tweet layer. }
          \label{layer12_comparison}
     \end{subfigure}
     \caption{Comparing the {smoothed (using regression)} degree distribution of the different layers of the TWITT network with our synthetic network and other standard networks on a logarithmic scale.}
     \label{TWITT_comparison}
 \end{figure*}
\textcolor{black}{In our literature survey we have not found existing methods for generating a network with diverse type of nodes and multiple layers. Hence, we compare the structural properties of the generated HLN by Algorithms \ref{algo1} and \ref{algo2} with the existing homogeneous models. The homogeneous models used are the {Barabási–Albert} (BA) model \cite{doi:10.1126/science.286.5439.509}, {Erdos–Rènyi} (ER) model \cite{Erdos:1959:pmd}, {Internet as a graph} \cite{5586438}, and {G(n,m)} Random graph. In each of the Figures \ref{layer1_comparison}, \ref{layer2_comparison}, \ref{layer12_comparison}, we have compared the degree distribution of the largest component of different layers of the  TWITT network with graphs generated from the aforementioned models as well as proposed synthetic network.} We have generated each of these networks with number of nodes ranging from $10000$ all the way to $40000$ and reported the degree distribution which is most comparable to the TWITT network for each comparing methods. For the ER model we use $p$ value of $0.1-0.9$ and report the best results. For the BA model we vary $m$ from $2-5$ and report the best results (we donot use a higher $m$ value as it does not produce comparable results). 
\textcolor{black}{As we can see from the Figures \ref{layer1_comparison}, \ref{layer2_comparison} and \ref{layer12_comparison} generated HLN is very close in replicating the degree distribution of the actual TWITT dataset in all the layers when compared to other modelling algorithms. It must be noted that we generated an HLN with two layers for representing the real Twitter network with $L = \{1, 2\}, R_L(1) = R_L(2) = \{t_1, t_2\}$. {For modelling the tweet-tweet and user-user layer with our HLN we use $20000$ nodes with $\alpha = \beta = 0.5$ and $m=3$, we use $m =4$ for the user tweet layer. The plots we use in Figure \ref{TWITT_comparison} are regression plots with log scale to clearly distinguish between the degree distribution of each model. } It must be noted our proposed network does not have any node having a degree less than the thresholds defined in $m$, similar to a BA network. When we compare our results with other models we see that our model is consistent  across the layers which shed light on the generic nature of our model. We have considered the interlayer as well as the intralayer degree of a node for preferential attachment, and from the Figures \ref{layer1_comparison} - \ref{layer12_comparison} it is evident that it well describes a real-world network. }
\begin{table*}[t]
\sffamily
\centering
\caption{Centrality measure and clustering co-efficient of \textbf{EATN} when compared to our generated HLN}
\resizebox{\textwidth}{!}{
\begin{tabular}{|l|l|l|l|l|l|l|l|}
\hline
\rule{0pt}{2ex}%  EXTRA vertical height
\textbf{Dataset} & \textbf{Nodes} & \textbf{Edges} & \textbf{Degree}  & \textbf{Betweenness} & \textbf{Avg CC} & \textbf{Avg Triangles/Node} & \textbf{Triangles}  \\\hline
\rule{0pt}{2ex}%  EXTRA vertical height
\textbf{EATN} & 55    & 97    & 0.06027 & 0.02592    & 0.45824    & 2.22683        & 62.62162   \\ \hline
\rowcolor{gray!15}
\rule{0pt}{2ex}%  EXTRA vertical height
\textit{\textbf{\textcolor{black}{Generated HLN ($\alpha$ = 1, $\beta$ = 0, $m$ = 2) }}} & 67    & 208   & 0.06147 & 0.02093    & 0.43264    & 3.59262        & 91.1111    \\ \hline
\end{tabular}
}
\label{tab:comparison_binball_HLNg}
\end{table*}
\begin{table*}[t]
    \sffamily
    \centering
    \caption{Comparison of Generated network with \textbf{Chemical} networks}
    \resizebox{\textwidth}{!}{
    \begin{tabular}{|l|l|l|l|l|l|l|l|l|l|}
    \hline
    \rule{0pt}{2ex}%  EXTRA vertical height
         \textbf{Datasets} & \textbf{Nodes} & \textbf{Edges} & \textbf{Density} & \textbf{Avg Degree} & \textbf{Assortativity} & \textbf{Triangles} & \textbf{Avg Triangles/Node} & \textbf{Avg CC} & \textbf{Clique Number} \\ \hline  
        %\textbf{$\alpha$ = 0.6, $\beta$ = 0.6} & ~ & ~ & ~ & ~ & ~ & ~ & ~ & ~ & ~ \\ \hline
        \rule{0pt}{2ex}%  EXTRA vertical height
         ENZYMES-g272 &  44 &  156 &  0.165 &  3.55 &  -0.05 &  47  &  1.07 &  0.23 &  2 \\ 
        ENZYMES-g366 & 42 & 152 & 0.176 & 3.62 & -0.03 & 48 & 1.14 & 0.23 & 1 \\ 
        \rowcolor{gray!15}\textbf{\textit{Generated HLN ($\alpha$ = 0.6, $\beta$ = 0.6, $m$ = 2)}} & 54 & 155 & 0.108 & 4.74 & -0.06 & 48 & 1.67 & 0.14 & 4 \\ \hline
        \rule{0pt}{2ex}%  EXTRA vertical height
        %$\alpha$ = 0.7, $\beta$ = 0.7 & ~ & ~ & ~ & ~ & ~ & ~ & ~ & ~ & ~ \\ \hline
        %\rule{0pt}{2ex}%  EXTRA vertical height
        ENZYMES-g392 & 48 & 178 & 0.158 & 3.71 & -0.02 & 74 & 1.54 & 0.26 & 1 \\ 
        ENZYMES-g117 & 46 & 180 & 0.174 & 3.91 & -0.03 & 59 & 1.28 & 0.19 & 2 \\ 
        \rowcolor{gray!15}\textbf{\textit{Generated HLN ($\alpha$ = 0.7, $\beta$ = 0.7, $m$ = 2)}} & 60 & 173 & 0.098 & 4.77 & -0.05 & 65 & 2.25 & 0.17 & 4 \\ \hline
        \rule{0pt}{2ex}%  EXTRA vertical height
        %$\alpha$ = 0.8, $\beta$ = 0.8 & ~ & ~ & ~ & ~ & ~ & ~ & ~ & ~ & ~ \\ \hline
        %\rule{0pt}{2ex}%  EXTRA vertical height
        ENZYMES-g526 & 58 & 220 & 0.133 & 3.79 & -0.02 & 66 & 1.14 & 0.21 & 1 \\ 
        ENZYMES-g527 & 57 & 214 & 0.134 & 3.75 & -0.03 & 80 & 1.4 & 0.27 & 2 \\ 
        \rowcolor{gray!15}\textbf{\textit{Generated HLN ($\alpha$ = 0.8, $\beta$ = 0.8, $m$ = 2)}} & 74 & 215 & 0.080 & 4.81 & -0.028 & 70 & 2.84 & 0.16 & 4 \\ \hline
        \rule{0pt}{2ex}%  EXTRA vertical height
        %$\alpha$ = 0.9, $\beta$ = 0.9 & ~ & ~ & ~ & ~ & ~ & ~ & ~ & ~ & ~ \\ \hline
        %\rule{0pt}{2ex}%  EXTRA vertical height
        ENZYMES-g349 & 64 & 236 & 0.117 & 3.69 & 0.00 & 78 & 1.22 & 0.24 & 2 \\ 
        ENZYMES-g103 & 59 & 230 & 0.134 & 3.9 & -0.03 & 73 & 1.24 & 0.22 & 2 \\ 
        \rowcolor{gray!15}\textbf{\textit{Generated HLN ($\alpha$ = 0.9, $\beta$ = 0.9, $m$ = 2)}} & 80 & 233 & 0.074 & 4.82 & -0.109 & 82 & 3.07 & 0.17 & 4 \\ \hline
        \rule{0pt}{2ex}%  EXTRA vertical height
        %$\alpha$ = 1.0, $\beta$ = 1.0 & ~ & ~ & ~ & ~ & ~ & ~ & ~ & ~ & ~ \\ \hline
        %\rule{0pt}{2ex}%  EXTRA vertical height
        ENZYMES-g295 & 123 & 278 & 0.037 & 2.26 & 0.00 & 6 & 0.05 & 0.01 & 1 \\ 
        ENZYMES-g296 & 125 & 282 & 0.036 & 2.26 & 0.00 & 2 & 0.02 & 0.01 & 1 \\ 
        \rowcolor{gray!15}\textbf{\textit{Generated HLN ($\alpha$ = 1.0, $\beta$ = 1.0, $m$ = 2)}} & 124 & 258 & 0.034 & 4.16 & -0.436 & 0 & 0.00 & 0.00 & 2 \\ \hline
    \end{tabular}
    }
    \label{mrTable1}
\end{table*}
\begin{table*}[]
    \sffamily
    \centering
    \caption{Comparison of Generated network with \textbf{Biological} networks}
    \resizebox{\textwidth}{!}{
    \begin{tabular}{|l|l|l|l|l|l|l|l|l|l|}
    \hline
    \rule{0pt}{2ex}
          \textbf{Datasets} & \textbf{Nodes} & \textbf{Edges} & \textbf{Density} & \textbf{Avg Degree} & \textbf{Assortativity} & \textbf{Triangles} & \textbf{Avg Triangles/Node} & \textbf{Avg CC} & \textbf{Clique N} \\ \hline
          \rule{0pt}{2ex}%  EXTRA vertical height
        Bio-yeast-protein-inter & 1846 & 4406 & 0.003 & 4.000 & -0.160 & 72 & 0.110 & 0.050 & 6 \\ 
        \rowcolor{gray!15}\textbf{Generated HLN ($\alpha$ = 0.7, $\beta$ = 0.7, $m$ = 2)} & 1978 & 5927 & 0.003 & 4.993 & -0.800 & 60 & 0.091 & 0.024 & 4 \\\hline
        %\textbf{$\alpha$ = 0.8, $\beta$ = 0.5} & ~ & ~ & ~ & ~ & ~ & ~ & ~ & ~ & ~ \\ \hline
        \rule{0pt}{2ex}%  EXTRA vertical height
        bio-DM-HT & 2989 & 4660 & 0.001 & 3.118 & -0.090 & 59 & 0.059 & 0.010 & 3 \\ 
        \rowcolor{gray!15}\textbf{Generated HLN ($\alpha$ = 0.8, $\beta$ = 0.5, $m$ = 2)} & 1508 & 4517 & 0.004 & 4.991 & -0.120 & 68 & 0.135 & 0.019 & 4 \\ \hline
        %$\alpha$ = 0.9, $\beta$ = 0.7 & ~ & ~ & ~ & ~ & ~ & ~ & ~ & ~ & ~ \\ \hline
        \rule{0pt}{2ex}%  EXTRA vertical height
        bio-grid-mouse & 1450 & 3272 & 0.003 & 4.000 & -0.150 & 120 & 0.248 & 0.030 & 7 \\ 
        \rowcolor{gray!15}\textbf{Generated HLN ($\alpha$ = 0.9, $\beta$ = 0.7, $m$ = 3)} & 1488 & 5939 & 0.005 & 6.983 & -0.190 & 110 & 0.222 & 0.028 & 5 \\ \hline
        %$\alpha$ = 0.7, $\beta$ = 0.7 & ~ & ~ & ~ & ~ & ~ & ~ & ~ & ~ & ~ \\ \hline
    \end{tabular}
    }
    \label{mrTable2}
\end{table*}
\begin{table*}[]
\sffamily
\centering
\caption{Comparison of Generated network with \textbf{Miscellaneous} networks}
\label{mrTable3}
\resizebox{\textwidth}{!}{%
\begin{tabular}{|l|l|l|l|l|l|l|l|l|l|}
\hline
\rule{0pt}{2ex}
          \textbf{Datasets} & \textbf{Nodes} & \textbf{Edges} & \textbf{Density} & \textbf{Avg Degree} & \textbf{Assortativity} & \textbf{Triangles} & \textbf{Avg Triangles/Node} & \textbf{Avg CC} & \textbf{Clique N} \\ \hline 
        \rule{0pt}{2ex}%  EXTRA vertical height
        cryg2500 & 2500 & 9849 & 0.003 & 7.000 & 0.600 & 400 & 0.400 & 0.010 & 4 \\
        Watt-2 & 1856 & 9694 & 0.006 & 10.000 & -0.070 & 838 & 0.930 & 0.020 & 4 \\ 
        \rowcolor{gray!15}\textbf{\textit{Generated HLN ($\alpha$ = 1.0, $\beta$ = 0.9, $m$ = 3)}} & 2441 & 9751 & 0.003 & 7.989 & 0.020 & 708 & 0.870 & 0.015 & 4 \\\hline
        \rule{0pt}{2ex}%  EXTRA vertical height
        PTC-MR & 4915 & 10108 & 0.001 & 4.000 & -0.300 & 180 & 0.109 & 0.000 & 3 \\
        PTC-FM & 4925 & 10110 & 0.001 & 4.000 & -0.300 & 168 & 0.102 & 0.000 & 3 \\ 
        \rowcolor{gray!15}\textbf{\textit{Generated HLN ($\alpha$ = 1.0, $\beta$ = 0.9, $m$ = 3)}} & 4059 & 10236 & 0.001 & 5.044 & -0.252 & 203 & 0.150 & 0.000 & 2 \\\hline
        \rule{0pt}{2ex}%  EXTRA vertical height
        M80PI-n1 & 4028 & 8066 & 0.001 & 4.000 & -0.140 & 0 & 0.000 & 0.000 & 2 \\
        S80PI-n1 & 4028 & 8066 & 0.001 & 4.000 & -0.140 & 0 & 0.000 & 0.000 & 2 \\ 
        \rowcolor{gray!15}\textbf{\textit{Generated HLN ($\alpha$ = 0.5, $\beta$ = 0.5, $m$ = 2)}} & 3909 & 7662 & 0.001 & 3.920 & -0.168 & 0 & 0.000 & 0.000 & 2 \\\hline
        \rule{0pt}{2ex}%  EXTRA vertical height
        Mutag & 3371 & 7442 & 0.001 & 4.000 & -0.260 & 0 & 0.000 & 0.000 & 2 \\ 
        \rowcolor{gray!15}\textbf{\textit{Generated HLN ($\alpha$ = 0.5, $\beta$ = 0.5, $m$ = 2)}} & 3822 & 7338 & 0.001 & 3.840 & -0.154 & 0 & 0.000 & 0.000 & 2 \\\hline
        \rule{0pt}{2ex}%  EXTRA vertical height
        bn-mouse-kasthuri-graph-v4 & 1029 & 1700 & 0.003 & 3.000 & -0.220 & 0 & 0.000 & 0.000 & 7 \\
        \rowcolor{gray!15}\textbf{\textit{Generated HLN ($\alpha$ = 0.8, $\beta$ = 0.5, $m$ = 2)}} & 807 & 1578 & 0.005 & 3.911 & -0.259 & 0 & 0.000 & 0.000 & 2 \\ \hline
        \rule{0pt}{2ex}%  EXTRA vertical height
        bibd-15-3 & 455 & 1364 & 0.013 & 5.000 & -0.630 & 166 & 1.094 & 0.010 & 4 \\ 
        \rowcolor{gray!15}\textbf{\textit{Generated HLN ($\alpha$ = 1.0, $\beta$ = 0.0, $m$ = 3)}} & 481 & 1436 & 0.012 & 5.971 & -0.632 & 180 & 1.123 & 0.047 & 4 \\ \hline
        \rule{0pt}{2ex}%  EXTRA vertical height
        lpi-bgdbg1 & 629 & 1579 & 0.008 & 5.000 & -0.040 & 159 & 0.658 & 0.010 & 3 \\
        \rowcolor{gray!15}\textbf{\textit{Generated HLN ($\alpha$ = 1.0, $\beta$ = 0.0, $m$ = 3)}} & 519 & 1550 & 0.012 & 5.973 & -0.015 & 143 & 0.827 & 0.027 & 4 \\ \hline
        \rule{0pt}{2ex}%  EXTRA vertical height
        Ia-crime-moreno & 829 & 1474 & 0.004 & 3.000 & -0.160 & 57 & 0.106 & 0.010 & 3 \\ 
        \rowcolor{gray!15}\textbf{\textit{Generated HLN ($\alpha$ = 1.0, $\beta$ = 0.0, $m$ = 2)}} & 763 & 1443 & 0.005 & 3.782 & -0.230 & 0 & 0.000 & 0.010 & 2 \\ \hline
\end{tabular}%
}
\end{table*}
\subsection{Multilayer Network Generation}
To demonstrate the capabilities of HMN for generating multilayer networks we have tried to model the European air transportation network \cite{Cardillo2013} abbreviated as
EATN. The air transportation network is also a multiplex network having 37 different layers with each layer representing different airlines of Europe. 

\subsubsection{Comparative Results}
We have generated 10 layers of the air transportation network using our HLN generating algorithm (with $L = \{1, 2, ..., 37\}, R_L(i) = \{t_1\}$) and compared the average centrality measures of the air transportation network
with our generated networks. The results are shown in Table \ref{tab:comparison_binball_HLNg}. It must be noted that the results presented in the table are averaged over the nodes of each of the layers in the multiplex network. The column Triangles/node denotes the average number of triangles any node participates in averaged over all the layers.
%-------------------------------------------
\subsection{Homogeneous Network Generation}
%\subsection{\textcolor{black}{Generating a multiplex network}}
We have successfully shown the generation capabilities of proposed {Algorithm \ref{algo1}} for generating real-world HLN and multilayer networks. The algorithm, however, can also be used to generate homogeneous networks (with $L = \{1\}$ and $R_L(1) = \{t_1\}$). In order to show the generalization capability of our proposed algorithm, we have tried to generate networks belonging to different domains like small molecule datasets (PTC-MR, PTC-FM), biological networks (bio-DM-HT, bio-grid-mouse, Bio-yeast-protein-inter), networks of nitroaromatic compounds (Mutag), crystal growth eigenmode graphs (cryg2500),  combinatorial problems (bibd-15-3), computational fluid dynamics graph (Watt-2), linear programming problems (lpi-bgdbg1), eigenvalue model reduction problems (M80PI-n1, S80PI-n1), chemical datasets (ENZYMES-g272, g366, g392, g117, g526, g527, g349, g103, g295, g296), brain networks (bn-mouse-kasthuri-graph-v4) and even crime dataset (Ia-crime-moreno). We have collected the networks from \cite{nr}. Networks belonging to different domains have different structural properties such as degree, density, centrality measure, number of triangles and assortativity, etc. We have compared our generated network with the existing networks on such structural parameters. The results are shown in Table \ref{mrTable1}, \ref{mrTable2}, \ref{mrTable3}. We have considered other structural measures apart from centrality measures. This includes density, average degree, assortativity, triangle counts, clustering co-efficient, and the number of max cliques. It must be noted that in all the tables, CC denotes clustering coefficient.

{Table \ref{mrTable1} models Chemical Networks using Algorithm 1 with varying parameters. Except ENZYMES-295 and g296 all other networks are single layer and hence the value of $L$ for them are kept $1$. For ENZYMES-g272 and g366, we use $\alpha = \beta = 0.6$, and $m = 2$. ENZYMES-g392 and g117 use $\alpha = \beta = 0.7$, while ENZYMES-526 and g527 use $\alpha = \beta = 0.8$, both with $m = 2$. For ENZYMES-349 and g103, $\alpha = \beta = 0.9$ is used. All single-layer networks have arbitrary $R_L$ values. ENZYMES-295 and g296 are modeled with $\alpha = \beta = 1.0$, $m = 2$, and $L = 2$, incorporating inter-layer connections, which inherently have zero triangles. In this case $R_L$ is a function that assigns nodes to a layer uniformly randomly.}

{In Table \ref{mrTable2}, we model biological networks with varying parameters. For the Bio-yeast-protein-inter network, we use $\alpha = \beta = 0.7$ and $m = 2$. The bio-grid-mouse network is modeled with $\alpha = 0.9, \beta = 0.7$ and $m = 3$. For the bio-DM-HT network, we use $\alpha = 0.8$, $\beta = 0.5$, and $m = 2$. All networks have $L = 1$ with arbitrary $R_L$. Increasing $m$ by 1 increases edge count while maintaining other network properties.}

{Table \ref{mrTable3} compares various network gneerated by our methods with the network from other domains. Computational graphs (cryg2500, Watt-2) and molecule datasets (PTC-MR, PTC-FM) are generated with $\alpha = 1.0$, $\beta = 0.9$, $m = 3$, and $L = 1$. Graphs used in eigenvalue model reduction problems (M80PI-n1, S80PI-n1) and nitroaromatic compound (Mutag) problem is obtained by using $\alpha = \beta = 0.5$, $m = 2$, and $L = 2$, incorporating inter layer connections. Brain networks (bn-mouse-kasthuri-graph-v4) are modeled with the same parameters, totaling 807 nodes across both layers. Combinatorial problem networks (bibd-15-3) and linear programming networks (lpi-bgdbg1) use $\alpha = 1.0$, $\beta = 0.0$, $m = 3$, and $L = 1$. Crime dataset networks (la-crime-moreno) are represented with $\alpha = 1.0$, $\beta = 0.0$, $m = 2$, and $L = 1$. The total node counts for all networks are reported in Table \ref{mrTable3}. In all these cases where the number of layers $L$ = $2$, $R_L$ assigns nodes to a layer uniformly randomly.}

\textit{One may \textbf{note}} that the network generation algorithm proposed in this paper is designed to have unknown node-correspondence (UNC) method \cite{Tantardini2019} of generation that does not use the number of nodes and edges as identity of a network. This is the reason we did not use the same number of nodes for most of the datasets. The reason behind this is that we wanted to focus more on generating a network with same structural properties with a comparable number of nodes. It must also be noted that there is a randomness in the selection of neighbours of a node (i.e. we use a uniform distribution for selecting candidates with the same degree) which results in non-determinism, i.e. network properties may vary slightly even with the same parameter values.

% \begin{figure*}
%      \centering
%      \begin{subfigure}[b]{0.45\textwidth}
%          \includegraphics[width=\textwidth]{air_binball_HLNg4.png}
%          \caption{}
%          %\label{fig:y equals x}
%      \end{subfigure}
%      %\hfill
%      \begin{subfigure}[b]{0.45\textwidth}
%          \includegraphics[width=\textwidth]{air_binball_HLNg5.png}
%          \caption{}
%          %\label{fig:three sin x}
%      \end{subfigure}
%      %\hfill
%      \begin{subfigure}[b]{0.45\textwidth}
%          \includegraphics[width=\textwidth]{air_binball_HLNg6.png}
%          \caption{}
%          %\label{fig:three sin x}
%      \end{subfigure}
%      %\hfill
%      \begin{subfigure}[b]{0.45\textwidth}
%          \includegraphics[width=\textwidth]{air_binball_HLNg7.png}
%          \caption{}
%          %\label{fig:three sin x}
%      \end{subfigure}
%      \caption{The figures compare the degree distributions of randomly selected layers from EATN with networks generated from BINBALL and HLNG.}
%      \label{air_binball_HLNg}
% \end{figure*}
%\section{\textcolor{black}{Applications}}
%\subsection{\textcolor{black}{Generating existing heterogeneous networks}}
%results table 
%------------original place-------------------------
%-----------------------

\section{Discussion and Conclusion}
\label{conclude}
In this paper, we introduced a new model Hybrid Layered Network (HLN), which is a generalized model of network capable of representing any complex networks of type homogeneous, heterogeneous, multilayer and their combinations. We also defined different structural measures on HLN. We have proved that the set of all HLNs is a superset of the set of all homogeneous, heterogeneous, and multi-layered networks.
In addition, a parameterized algorithm is presented to generate an HLN synthetically. We show that the algorithm is able to generate a homogeneous, heterogeneous, and multilayered network by changing parameter values. 

{\textbf{Limitations:} In this work we only provided the generalized definitions of certain structural measures like degree centrality, betweeness centrality and closeness centrality. However, in network science there are many different structural properties defined for homogeneous, heterogeneous and multilayer graphs. This would be a good study to develop corresponding definitions for HLN as well. Some of these measure can be clustering co-efficient, triangles and cliques for HLN. The proposed HLN generation algorithm may not generate certain networks as it has limited number of parameters. Hence, it could only able to generate networks that is UNC with comparable node and edge counts. A better algorithm leveraging the generative capabilities of new edge GNN can be used in the future to generate synthetic graphs which can simultaneously match the numbers of nodes, edges and other structural properties of the network. Furthermore, the intra layer networks generated by our algorithm follow a scale free property. An extension can be made to generate specific degree distribution in the future. Finally, it is of interest to find the theoretical bounds for the networks generated by our algorithm.}

{Despite the aforementioned limitations, the network generated by the algorithm is generalized and can be tweaked by changing the parameter values for applications in certain areas where networks largely follow a scale-free property. These synthetic networks will open the opportunity to research with HLNs that is otherwise difficult to conduct due to the unavailability of the data set. With the availability of the services (e.g., Fediverse) on ActivityPub protocol, we expect the real-world HLN data will be available in near future. Although heterogeneous network data sets are available for research, to the best of our knowledge, there is no algorithm for generating a heterogeneous network and the proposed algorithm would encourage research with heterogeneous networks as well. Note that the proposed algorithm can only be used to generate an undirected HLN. However, we believe that with minor changes, we can generate directed HLNs as well. An important future research is to show that our proposed definition of HLN holds for a dynamic network or a signed network.}

Finally, with this work, we tried to open a new avenue of research with complex networks. While the theories developed will help further theoretical analysis and provide the basis of application, the synthetic network generation algorithm will provide the opportunity to develop applications with HLN.

\bibliography{archivx.bib}% Produces the bibliography via BibTeX.

\end{document}